\documentclass[pra,showpacs,twocolumn,superscriptaddress]{revtex4}

\usepackage{mathtools}
\usepackage{float}
\usepackage{graphicx}
\usepackage{epsfig}
\usepackage{dcolumn}
\usepackage{amsmath,amssymb,amsthm,amsfonts}
\usepackage{bm}
\usepackage{verbatim}

\newtheorem{theorem}{Theorem}
\newtheorem{definition}{Definition}
\newtheorem{lemma}{Lemma}

\let\originalleft\left
\let\originalright\right
\renewcommand{\left}{\mathopen{}\mathclose\bgroup\originalleft}
\renewcommand{\right}{\aftergroup\egroup\originalright}

\newcommand{\bra}[1]{\left\langle #1 \right|}
\newcommand{\ket}[1]{\left| #1 \right\rangle}
\newcommand{\braket}[2]{\left\langle #1 \middle| #2 \right\rangle}
\newcommand{\ketbra}[2]{\left|#1\middle\rangle\middle\langle#2\right|}

\newcommand{\abs}[1]{\left|#1\right|}

\newcommand{\comm}[2]{\left[#1,#2\right]}
\newcommand{\anti}[2]{\left\{#1,#2\right\}}
\newcommand{\M}[1]{\mathcal{#1}}

\newcommand{\ad}[1]{a^{\dagger}_{#1}}

\newcommand{\id}{\mathbb{I}}
\newcommand{\B}[1]{\mathbb{#1}}
\newcommand{\Tr}{\text{Tr}}

\begin{document}
\title{Quantumness of Correlations in Fermionic Systems} 
\author{Tiago Debarba} 
\affiliation{Universidade Tecnol\'ogica Federal do Paran\'a
(UTFPR), Campus Corn\'elio Proc\'opio, Avenida Alberto Carazzai
1640, Corn\'elio Proc\'opio, Paran\'a 86300-000, Brazil}
\email{debarba@utfpr.edu.br}
\author{Reinaldo O. Vianna}
\affiliation{Departamento de F\'{\i}sica - ICEx - Universidade Federal de Minas Gerais,
Av. Pres.  Ant\^onio Carlos 6627 - Belo Horizonte - MG - Brazil - 31270-901.}
\author{Fernando Iemini}
\affiliation{Abdus Salam ICTP, Strada Costiera 11, I-34151
Trieste, Italy}
\email{fernandoiemini@gmail.com}
\date{\today}
\begin{abstract}

 We present a new approach for the quantification of quantumness of correlations in 
fermionic systems.
We study the Multipartite Relative Entropy of Quantumness in such systems, and show how the symmetries 
in the states can be used to obtain analytical solutions.
Numerical evidences about the uniqueness of such solutions are also presented. 
Supported by these results, 
we show that the minimization of the Multipartite Relative Entropy of Quantumness, over certain choices of its 
modes multipartitions, reduces to the notion of Quantumness of Indistinguishable Particles. 
By means of an {\it activation protocol}, we characterize the class of states without quantumness of correlations.
As an example, we calculate the dynamics of 
quantumness of correlations for a purely dissipative system, whose stationary states exhibit interesting topological non-local correlations. 

\end{abstract}
\maketitle

\section{Introduction}
The understanding of quantum correlations in systems of indistinguishable particles, especially fermions,
is paramount for the development of materials supporting the new technologies of Quantum Information
and Quantum Computation \cite{chuang,review.Nayak_2008}. 
The subtle notion of entanglement of indistinguishable particles has been investigated by many authors back in the
2000's, with introduction of diverse seminal ideas like entanglement of modes  \cite{zanardi02},
 and entanglement of particles \cite{wiseman03,eckert02}.  Such ideas have been  applied as new tools in the investigation of many-body system properties 
 \cite{amico08}, including the characterization of quantum phase transitions \cite{iemini15,amico08}.

From a mathematical viewpoint, the difficulty of understanding entanglement of indistinguishable particles stems from  the absence of a tensor product structure in the
Fock space, whereas the concept of entanglement is based on the non-separability, with respect to the tensor product, 
of a global state of identifiable subsystems. 
Friis {\it et al.} \cite{friis13} and Balachandran {\it et al.} \cite{balachandran}
 have suggested interesting mathematical approaches to circumvent this obstacle. 
 Our own efforts in this problem started with the proposal of entanglement witnesses for systems of indistinguishable particles \cite{iemini13}, followed by
 appropriate adaptations of   entropy of entanglement and negativity \cite{iemini_ComputableMeasures}.  
 More recently, extending the concept of {\it quantumness of correlations}  \cite{henderson,zurek2001,modiprl2010,debarbapra} 
 to the realm of indistinguishable particles, we introduced the concept of 
 {\it quantumness of correlations of indistinguishable particles} \cite{iemini14}, which allowed us to devise a measurement procedure 
 (dubbed an {\it activation protocol}) \cite{pianiprl,streltsov11} that determines the class of states without quantumness of correlations. 
  In this activation protocol the quantumness of correlations of the system 
  manifests itself as the smallest amount of bipartite entanglement created between system and measurement ancilla, 
  during the local measurement protocol. 
  
In order to recover the tensor product structure and the separability 
of the subsystems, one may use the isomorphism of the Fock 
space to a Hilbert space of distinguishable modes. 
 This approach \cite{zanardi02,wiseman03,benatti10,banulus07}
 allows one to employ all the tools commonly
 used in distinguishable quantum systems for
 the analysis of the correlation between the modes of the system.
  The two notions of correlation in such systems, namely modes or particle
   correlation, follows naturally depending 
on the particular situation under scrutiny. For example,
the correlations in eigenstates of a many-body Hamiltonian
could be more naturally described by particle entanglement,
whereas certain quantum information protocols could prompt
a description in terms of entanglement of modes.


{As the composed space of modes has a tensor product structure,
 given  the modes are distinguishable, and there exists an isomorphism connecting 
the Hilbert space of modes and the Fock space of particles,
 some questions naturally  arise: What is the relation between a correlated system of distinguishable modes 
and the correlations of indistinguishable particles?}
Is it possible to characterize or quantify the quantum correlations of particles by means of their mode quantum correlations? 
Can we characterize the set of uncorrelated states of indistinguishable particles out of  the  description of  distinguishable modes? 
In this work we present a new approach for 
the  quantumness of correlations of indistinguishable 
particles described by the minimization of the modes 
  representation of the particles system. 
  We show that, in the single particle partitioning, both notions are equivalent. 
    
  This work is organized as follows. 
  In Sec. \ref{form}, we present the formalism of Fock space for fermions, and its isomorphism with the Hilbert space of modes. 
  In Sec. \ref{quant.lema1}, we present the local measurement 
  formalism and introduce a quantifier of quantumness of correlations based on the 
  local disturbance. We also present one of the main results of this work, proving 
  that for quantum states 
with a given arbitrary symmetry, optimal local projective measurements - which 
minimize the local disturbance - are symmetric.
In Sec. \ref{modes.part.quant.}, we discuss 
 the connection between quantumness of correlations of modes and quantumness of 
correlations of particles for single particle modes. 
{In Sec. \ref{act.prot.sec}, we characterize the set of  fermionic states  without quantumness of correlations.
This results is obtained from an   
activation protocol for a system of L  modes.  	}
In Sec. \ref{dissip.system},  we illustrate  our results studying the dynamics of quantumness of correlations of a purely dissipative system. 
In Sec.  \ref{num.res},  we analyze, numerically, the quantumness of correlations of more general symmetric quantum states and their respective local projectors. 
Conclusions are presented in Sec. \ref{conc}.

\section{Formalism}\label{form}
The space of quantum states for systems of indistinguishable
fermions (bosons) is given by the anti-symmetric (symmetric)
Hilbert-Schmidt subspace. Along all the paper, 
for simplicity, we will focus our calculations on the fermionic case, despite 
all results could be easily 
translated to the bosonic case. Formally, the  quantum states for a
fermionic system with $L$ modes are in the  the Fock space
($\M{F}_L$), namely, 
\begin{equation*}
\M{F}_L = \ket{vac}\bra{vac} \oplus \mathcal{A}(\mathcal{H}^L)
\oplus \mathcal{A}(\mathcal{H}^L \otimes \mathcal{H}^L)
\oplus ... \oplus \mathcal{A}(\mathcal{H}^{L^{\otimes L}}),
\end{equation*}
\textit{i.e.}, the direct sum of anti-symmetric Hilbert spaces ($
\mathcal{A}(\mathcal{H}^{L^{\otimes N}})$) with fixed
$N=0,1,...,L$ fermions, with $| vac \rangle$ being the vacuum state.
Recall that the direct sum is defined as:
\begin{equation}
A \oplus B = \left[\begin{array}{cc}
A & 0 \\
0 & B 
\end{array} \right].
\end{equation}
The dimension of the Fock space is given by:
\begin{equation}\label{eq-rov11}
N_F^L = \sum_{i=0}^L \frac{L!}{(L-i)! i!} = 2^L.
\end{equation}
A basis can be generated from a set of
single-particle fermionic operators $\{a_j\}_{j=1}^L$ for the
$L$ modes,
satisfying the canonical
anti-commutation relations:
\begin{equation}\label{eq-rov10}
\{a_i,a_j^\dagger\}=\delta_{ij}, \,\,\, \{a_i,a_j\}=0.
\end{equation}
We represent the basis with the short notation,
\begin{equation}\label{eq:basis.Fock}
| \ad{\vec{j}} \rangle \equiv (\ad{1})^{j_1} (\ad{2})^{j_2}
...\, (\ad{L})^{j_L} \ket{vac},
\end{equation}
with $\vec{j} = (j_1,...,j_L)$, and $j_i = 0 (1)$ denoting an
empty (occupied) mode.

By means of the \textit{occupation-number representation}, the Fock space can be associated to a 
Hilbert space of $L$ qubits with a $2^L$-dimensional basis, $\{\ket{j_1j_2\ldots j_L}\}$, where $j_\ell$ is $0$ or $1$ for unoccupied or occupied modes, respectively.

Formally, these two equivalent representations are related by the 
 following isomorphism $\Lambda$:
\begin{eqnarray}\label{isomorphism}
\Lambda:\M{F}_L & \longleftrightarrow & 
\mathcal{H}^{2}_1 \otimes \cdots \otimes \mathcal{H}^{2}_L   \\
| \ad{\vec{j}} \rangle \equiv (\ad{1})^{j_1} 
...\, (\ad{L})^{j_L} \ket{vac} &\longleftrightarrow & \ket{j_1}
\otimes ... \otimes \ket{j_L} \equiv | \vec{j}\rangle, \nonumber
\end{eqnarray}
which maps the Fock space into the space of $L$ distinguishable
modes (qubits in the fermionic case). 
We will denote
hereafter as ``configuration
representation (modes representation)'' the left (right) side of
the previous equation. 
Notice that, in principle, the basis $\{a_j\}_{j=1}^L$ of
fermionic operators is arbitrarily chosen, and the previous
isomorphism is completely dependent on it. 
Therefore different modes representations can be obtained by means of a unitary 
transformation on the Fock space $U\in\M{U}(\M{F}_L)$, named Bogoliubov transformations,
\begin{equation}\label{loc.unit}
\tilde{f}^{\dagger}_k = \sum_j U_{kj} \ad{j}, 
\end{equation}   
where $\{\tilde{f}^{\dagger}_k\}_{k=1}^{L}$ and $\{\ad{j}\}_{j=1}^{L}$ are single particle fermionic operators.


\section{Quantumness of Correlations}\label{quant.lema1}

In this section we will first introduce the description of local
measurements
on the modes. In this context, we will introduce a
quantifier
of quantumness of correlations based on the smallest 
disturbance created by local measurements. We will then be ready
to present
one of our main results, proving that for quantum states 
with a given arbitrary symmetry, 
the optimal local projective measurement, \textit{i.e.,} the
local projectors
which creates the smallest disturbance in the state, is that
which shares the
same symmetry with the quantum state.

Let us first formally define the concept of projective
measurement. Given a general bi-partition of the $L$ modes
as $j_A$/$j_B$, with $j_B = L-j_A$, a general set
of local projective measurements $\{ \hat{\Pi}^{(j_A)}_m\}_{m} $
acting on $j_A$ is defined as:
\begin{eqnarray}
\hat{\Pi}_m^{(j_A)} \hat{\Pi}_{n}^{(j_A)} &=& \delta_{m,n}
\hat{\Pi}_m^{(j_A)}, \\
\sum_m \hat{\Pi}_m^{(j_A)} &=& \id.
\end{eqnarray}

The quantum state $\rho\in\M{D}(\M{H}^{2j_A}\otimes\M{H}^{2j_B})$, 
where $\mathcal{D}$
denotes the set of all positive semi-definite operators, 
after such measurement
is described as a local dephased state, given by:
\begin{eqnarray}
 \Pi^{(j_A)} (\rho) &=& \sum_m \hat{\Pi}_m^{(j_A)} \rho
\hat{\Pi}_m^{(j_A)} \nonumber \\
&=& \sum_m \Tr_{j_A}(\hat{\Pi}_m^{(j_A)} \rho) \,
\hat{\Pi}_m^{(j_A)} \nonumber \\
&=& \sum_m p_{\Pi_m}^{(j_A)} \,\hat{\Pi}_m^{(j_A)}\otimes
\sigma_m^{(j_B)},
\label{gen.projective.meas}
\end{eqnarray}
where 
\begin{equation}
p_{\Pi_m}^{(j_A)} = \Tr(\hat{\Pi}_m^{(j_A)} \rho),\qquad
\sigma_m^{(j_B)} = \frac{1}{p_{\Pi_m}^{(j_A)}}
\Tr_{j_A}(\hat{\Pi}_m^{(j_A)} \rho),
\end{equation}
with $\sigma_m^{(j_B)}$ representing the reduced state, which
might be a mixed state, in the complementary $j_B$ subspace.


The quantification of quantumness of correlations in a quantum
system
can be performed by means of the smallest disturbance created 
by a local measurement \cite{brodutch12}. Such disturbance could
be quantified
as the distance between the original state and the measured 
one \cite{fu08,debarbapra, taka, modiprl2010}. 
In this way we define the multipartite relative entropy of
quantumness
\cite{modiprl2010,opp}:

\begin{definition}[\textbf{Multipartite Relative Entropy of
Quantumness- MREQ}]\label{defimult}
Given an arbitrary multipartite quantum state  
$\rho \in \M{D}(\mathcal{H}_1 \otimes\cdots\otimes
\mathcal{H}_n)$ and
a set $\Sigma\subseteq\{1,2,\dots,n\}$, the Multipartite
Relative Entropy
of Quantumness on the subsystems $\Sigma$ is defined as :
\begin{equation}
Q^{\Sigma}(\rho) = \min_{\{\bigotimes_{j\in\Sigma}
 \Pi^{(j)}\} }
S\left(\rho\left|\right|\bigotimes_{j\in\Sigma}
\Pi^{(j)}(\rho)\right)
\end{equation}
where $ \Pi^{(j)}(\rho)$ is a local projective measurement map over the
$j$'th partition of $\rho$ and
$S(\rho||\sigma) = - S(\rho) - Tr(\rho \log \sigma)$ is the
relative entropy between $\rho$ and $\sigma$,
 with $S(\rho)$ the von Neumann entropy.
\end{definition}

Considering a bi-partition on the modes, $\rho\in\M{D}(\M{H}_{A}\otimes\M{H}_{B})$, and restricting the set $\Sigma = A$, the relative entropy is
the 
quantifier of quantumness of correlations corresponding to the  {\it One-Way
Work Deficit} \cite{opp}:
\begin{eqnarray}
Q^{(A)}(\rho) &=& \min_{\{  \Pi^{A}\} } S(\rho||
 \Pi^{A}(\rho)) \nonumber \\
&=&  
\min_{\{   \Pi^{A}\} } \left[ S(
 \Pi^{A}(\rho))-S(\rho) \right].
\label{defi2}
\end{eqnarray}

The computation of the above quantifiers is not a trivial task,
 in general,  due to the minimization over all local projective
measurements. The solution of
the above minimization might not even be unique; however, by
definition, it is always a set of
rank-1 local projectors
\cite{horodecki2005local}. One could try to use some information
of the
quantum state under analysis in order to solve the optimization, 
or at least
to restrict the search of the optimal projective measurement to
a subset.
In the context of ensemble of quantum states generated by a
symmetric group, it is known that the POVM which  optimizes the
accessible information is
also created by the symmetric group \cite{ban97}.
It is analogous to the quantum state discrimination problem, where the concerned states
and the optimal POVM have the same symmetry \cite{sasaki99,davies}. 
In Ref.\cite{dariano06},  it is shown that, under the 
action of a symmetric group in the probability space, the extremal covariant POVMs
are  in one-to-one correspondence
to  the convex set of  block diagonal
operators on the Hilbert space.
Therefore symmetries can be written as degenerated observables.
Following this approach, 
we will show here that, in the context of quantumness of
correlations, if the quantum state has a given symmetry, the
minimization task can be performed
restricted to symmetric projective measurements. 
This simplification follows from the following Lemma.
  
\begin{lemma}\label{lemma}
Consider a symmetry
$\hat{\Theta}\in\M{L}(\M{H}_1\oplus\cdots\oplus\M{H}_N)$,
where $\M{L}(\cdots)$  denotes
 the set of linear operators, 
$\theta_j$ are the eigenvalues of the symmetry and
$\M{H}_j$ their corresponding
block diagonal subspace.
Given a quantum state $\rho \in \mathcal{D}(\mathcal{{H}}_j=\M{{H}}_j^A\otimes \M{{H}}_j^B)$,
\textit{i.e., } a quantum state
with non trivial projection onto a single eigenvalue of the
symmetry,
there exists a set of {symmetric} local projective 
measurements $\{ \hat P_{\ell}^{(A)} \}_\ell $ 
that is a solution of the One-Way Work Deficit - Eq.\eqref{defi2}. This projective
measurement acts locally on the $A$ bi-partition,
{either} preserving the state symmetry,  
$ \hat P_{\ell}^{(A)} \rho \hat P_{\ell}^{(A)}  \in \mathcal{D}(\mathcal{{H}}_j)$,
 {or  annihilating the state, 
 $\hat P_{\ell}^{(A)} \rho \hat P_{\ell}^{(A)}  = 0 $,
  in case the symmetry subspace defined by $\hat P_{\ell}^{(A)}$ is orthogonal 
  to the state symmetry.}
\end{lemma}
According to the Lemma, the set of symmetric projective measurements
creates the smallest disturbance in the symmetric state. 
It explores the block diagonal structure of the
symmetry, in the sense that an
optimal symmetric projective measurement must preserve such
structure.
Consequently we have a great simplification of the optimization
problem in the 
computation of  the relative entropy of quantumness. 
As we will discuss in the next section, for some symmetries and modes partitions, 
there exists a unique 
projective measurement satisfying the Lemma conditions, thus
allowing us to compute the quantumness analytically.
The proof for the Lemma \ref{lemma} is given in the Appendix.

\textit{Example}. Let us present a simple example illustrating the previous discussion.
Assume a bipartite distinguishable system  ($A-B$) in a state with the following Schmidt decomposition:
\begin{equation} \label{eq-rov1}
\ket{\psi}=\sum_{i=1}^r \sqrt{p_i} \ket{a_i}\ket{b_i}.
\end{equation}
If $\{\ket{a_i}\}$ is an orthonormal basis in the Hilbert space of a qubit, while $\{\ket{b_i}\}$ is an orthonormal basis for a Hilbert space of arbitrary finite dimension, then $r$ is at most 2.
Consider a local projective measurement in the Schmidt basis,
\begin{equation}
 \Pi^A_1 = \ketbra{a_1}{a_1} \otimes \mathbb{I}_B,\quad  \hat \Pi^A_2 = \ketbra{a_2}{a_2} \otimes \mathbb{I}_B.
\end{equation}
Then we have:
\begin{eqnarray}
 \Pi^A(\rho) &=& \hat \Pi^A_1\rho \hat \Pi^A_1 + \hat \Pi^A_2 \rho \hat \Pi^A_2\ \nonumber \\
&=& p_1 \ketbra{a_1 b_1}{a_1 b_1} + (1-p_1) \ketbra{a_2 b_2}{a_2 b_2} 
\end{eqnarray}
with
\begin{eqnarray}
 \bra{\psi}(\ketbra{a_1}{a_1}\otimes \mathbb{I}_B)\ket{\psi} &=& p_1 \\
 \bra{\psi}(\ketbra{a_2}{a_2}\otimes \mathbb{I}_B)\ket{\psi} &=&(1-p_1)
\end{eqnarray}
and the relative entropy for such local measurement is given by,
\begin{equation}
S(\rho || \Pi^A (\rho) ) = h(p_1)=-p_1\log{p_1}-(1-p_1)\log(1-p_1).\nonumber
\end{equation}
Thus we have determined the entropy of entanglement of this system ($E(\rho) = S(\rho_A) = h(p_1)$, with
 $\rho_A = Tr_B(\rho)$)
 by means of local projective measurements. A lesson here is that it would be
easy to measure entanglement if we knew the Schmidt
basis.
Let us check what bound we would obtain measuring in an arbitrary local basis, namely:
\begin{eqnarray}\label{eq-rov5}
\ket{\tilde{a}_i}=\alpha_{1i}\ket{a_1}+ \alpha_{2i}\ket{a_2}, \\ \nonumber
\braket{\tilde{a}_i}{\tilde{a}_j}=\delta_{ij},
\end{eqnarray}
that is, 
\begin{equation}\label{eq-rov6}
\ket{\tilde{a}_i}=U\ket{a_i},
\end{equation}
where $U$ is an arbitrary unitary transformation.
Now, if we perform a local projective measurement in the basis $\{\ket{ \tilde{a}_i}\}$, we obtain:
\begin{eqnarray}\label{eq-rov7}
\tilde{p}_1=\bra{\psi}(\ketbra{\tilde{a}_1}{\tilde{a}_1}\otimes \mathbb{I}_B)\ket{\psi}, \\ \nonumber
\tilde{p}_1=\abs{\alpha_{11}}^2 p_1 + \abs{\alpha_{21}}^2 (1-p_1).
\end{eqnarray}
As the binary entropy is a concave function, we have:
\begin{eqnarray}\label{eq-rov8}
h(\tilde{p}_1)\geq \abs{\alpha_{11}}^2 h(p_1) + \abs{\alpha_{21}}^2 h(1-p_1), \\ \nonumber
\Rightarrow h(\tilde{p}_1)\geq h(p_1).
\end{eqnarray}
Therefore, a measurement in an arbitrary basis gives us an upper bound for the relative entropy, and
we could obtain the quantumness of correlations - One-Way Work Deficit - by a minimization over all different bases:
\begin{equation}\label{eq-rov9}
Q^A(\psi)=\min_U h(\tilde{p}_1).
\end{equation}

Now we will proceed to the case of a system of 
indistinguishable fermions, and discuss how 
some symmetries of the quantum state can be useful in determining the optimal local
measurements for the relative entropy.

We shall discuss the case of states with
definite parity. 
Consider the following projectors onto a single particle mode:
\begin{equation}\label{eq-rov14}
\hat \Pi_j=a_j^\dagger a_j, \,\,\, \hat \Pi_{\bar{j}}=a_j a_j^\dagger, \,\,\, \hat\Pi_j+ \hat\Pi_{\bar{j}}=1.
\end{equation}

Now an arbitrary Fock state with definite parity can be written as:
\begin{equation}\label{eq-rov15}
\ket{\psi}=\hat \Pi_j\ket{\psi}+ \hat \Pi_{\bar{j}} \ket{\psi}=\sqrt{p_j}\ket{\psi_{j}} + \sqrt{p_{\bar{j}} }\ket{\psi_{\bar{j}}},
\end{equation}
where we have defined the following state:
\begin{equation}\label{eq-rov16}
\ket{\psi_{j}}=\frac{\hat \Pi_j \ket{\psi}}{\sqrt{p_j}},\quad  p_j=\langle{\psi}|\hat \Pi_j|\psi\rangle,
\end{equation}
and analogously for $\ket{\psi_{\bar{j}}}$.
These projectors (Eq.\ref{eq-rov14}) define a bi-partition in the Fock space and 
directly determine the Schmidt decomposition for the state. $\ket{\psi_{j}}$ is in the subspace
where the mode $j$ is occupied, whereas $\ket{\psi_{\bar{j}}}$ is in the complementary subspace. 

Recalling the previous discussion for general states (Eq.\eqref{eq-rov1}), we easily conclude that the local
 measurement from these projectors determine the quantumness of correlations - One-Way Work Deficit - 
 for the state. Notice that the measurement preserves the symmetry of the state, according to 
 the Lemma, and the above rationale is valid for arbitrary states with parity symmetry.


 \section{Equivalence between Quantumness of Correlations of Modes and Particles}\label{modes.part.quant.}

In this section we will discuss the implication of
Lemma \ref{lemma} in the context
 of quantum states with fixed parity symmetry.
 In particular, we show that the notion of 
\textit{Quantumness between Indistinguishable 
Particles} can be recovered from quantumness between modes, by
 analyzing the minimum of the Multipartite Relative
Entropy
of Quantumness over certain choices of  modes multipartitions.

Since the quantumness of correlations is implicitly related to 
local measurements on composite systems, 
 the characterization, or even a proper definition, of the quantumness of correlations
  between \textit{indistinguishable particles} 
is a much subtler task. In such systems particles are no
longer accessible individually,
thus eliminating the usual notions of separability and local
measurements.
In Ref.\cite{iemini14},  the authors define such a notion of
quantumness of correlations
between 
indistinguishable particles
by means of the activation protocol, which 
relates the correlations to the 
smallest amount of entanglement created between system and
apparatus during a single
particle measurement. 
As  system and apparatus are always distinguishable, the
quantumness of correlations of
a indistinguishable system can be obtained from usual
distinguishable quantities.

Let us recall 
the measure of quantumness of correlations as proposed 
in Ref.\cite{iemini14}:

\begin{definition}[\textbf{Quantumness of correlations between
indistinguishable particles}]
Considering a system of particles described by the density
matrix $\rho\in\M{D}({\M{F}_{L}})$, the quantumness of
correlations between the
indistinguishable particles is defined by means
of the relative entropy, as follows \footnote{The index $\emptyset$ indicates that the local 
measurement is performed over each particle. This a multipartite analogous to the zero-way work deficit \cite{opp}.}:
\begin{equation}\label{quant.part.re}
\mathcal{Q}_p^{\emptyset}(\rho) = \min_{V\in\M{U}({\M{F}_{L}})} S(\rho||
 \hat \Pi^{V}(\rho) ),
\end{equation}
where \begin{equation}\hat{\Pi}^{V}(\rho) = \sum_{\{\vec{\ell}\}} \langle
a_{\vec{\ell}}| V^{\dagger} \rho V | a_{\vec{\ell}} \rangle \,
|a_{\vec{\ell}}\rangle \langle a_{\vec{\ell}}|,
\end{equation} with 
$\{ \ket{a_{\vec{\ell}}}\}$ being a single Slater determinant basis, as defined in Eq.\eqref{eq:basis.Fock},  and 
 $V$ a unitary transformation on Fock space, as defined in Eq.\eqref{loc.unit}. 
\end{definition}

An important consequence of the Lemma 
follows for quantum states with fixed parity symmetry.
Considering a
bi-partition of the modes between a single-particle mode
``$a_{j}$'' with the rest of the system, the set
of rank-$1$
local projective measurements onto such mode, preserving the symmetry of the 
state, reduces to a single 
possible set (and thus a solution for the One-Way Work Deficit):
\begin{eqnarray}\label{proj.meas}
\hat \Pi^{j}_{min} = \left\{ \hat \Pi^{j}_0 = a_{j} a_{j}^{\dagger}, \,
\hat \Pi^{j}_1 = a_{j}^{\dagger} a_{j} \right\}.
 \end{eqnarray}
In the case of a multi-partition of the system in $L$ subsystems, with each one
corresponding to a single-particle mode $a_j$, 
the set
of rank-$1$
local projective measurements onto such subsystems, 
preserving the symmetry of the 
state, reduces to:
\begin{equation}
\hat \Pi^{1\,2...L}_{min} = \left\{ \hat \Pi^{1}_\ell \otimes 
\hat \Pi^{2}_\ell \otimes \cdots \otimes \hat \Pi^{L}_\ell  \right\}, \quad (\ell=0,1),
\end{equation}
and thus represents  the constrained set solution for the MREQ.
 
With the previous considerations, we are now ready  to
 present one of the most important results of this work.
We relate
the notion of quantumness of correlations between
indistinguishable particles to
the MREQ minimized over all possible single-particle modes
partitions.
 
\begin{theorem}\label{teo1}
Given a system with $L$ modes, described by the 
state $\rho\in\M{D}(\M{H}_1^2\otimes\cdots\otimes\M{H}_L^2)$ with 
fixed parity symmetry (\textit{cf} Lemma), 
the minimization of the Multipartite Relative Entropy of
Quantumness, 
over all single-particle representations ``$\{a\}$'', is equal to the Quantumness
of correlations
 between indistinguishable particles: 
\begin{equation}\label{eq:quantumness.particles}
\M{Q}^{\emptyset}_p(\rho) = \min_{\{a\}}
Q_{a}^{1,\ldots,L}(\rho),
\end{equation}
where $a: \M{F}_L \leftrightarrow \M{H}^{2\otimes L}$ represents
 the isomorphism between Fock and modes space (Eq.\eqref{isomorphism}),  and
${1,\ldots,L}$
indicates that the measurement is performed locally over all
modes.
\end{theorem}
\begin{proof}
Given the relative entropy of quantumness for the modes,
perform a measurement over all of
them locally, in a representation $\Lambda$: 
\[Q^{1,\ldots ,L}_{a}(\rho) =\min_{\Pi^{a}}
S(\rho\Vert \Pi^{a}(\rho))=S(\rho\Vert
\Pi^{a}(\rho)),\]
where $\Pi^{a}(\rho)$ 
is the local measurement map over all modes in the
representation
$a$. By means of Lemma \ref{lemma},  the minimization over
all projective measurements
in this representation is restricted to only one projective
measurement map $\Pi^{a}$.
The projective measurement map acts on $\rho$ as: 
\begin{align*}
\Pi^{a}(\rho)&=\hat\Pi_{1}^{a}\otimes\cdots\otimes
\hat\Pi_{L}^{a}(\rho)\\
& = \sum_{\vec{l}}\hat\Pi_{\vec{l}} \rho \hat\Pi_{\vec{l}}\\
& =  \sum_{\vec{l}} \ketbra{a_{\vec{\ell}}}{a_{\vec{\ell}}} \rho \ketbra{a_{\vec{\ell}}}{a_{\vec{\ell}}},
\end{align*}
where $\ket{a_{\vec{\ell}}}=a_{\vec{l}}\ket{vac}=a_1^{l_1\dagger}\cdots a_L^{l_L\dagger}\ket{vac}$.
The transformation from the representation $\Lambda$ to  another
$\Lambda'$ can be performed
by means of a Bogoliubov transformation $V$, thus:
\[ f_{\vec{l}}^{\dagger}= V  a_{\vec{l}}^{\dagger},  \]
where ${V}$ is a unitary
operation on the particles. Therefore
minimizing  the quantumness of modes over all representations amounts 
to minimize it over all Bogoliubov
transformations $V$: 
\[
\min_{\Pi^{a}} S(\rho\Vert \Pi^{a}(\rho)) = \min_V
S(\rho\Vert \Pi^{V}(\rho)),
\]
where \[ \Pi^{V}(\rho) = \sum_l \bra{a_{\vec{\ell}}}V^{\dagger} \rho
V \ket{a_{\vec{\ell}}} \ketbra{a_{\vec{\ell}}}{a_{\vec{\ell}}}, \]
which  results in the measure of quantumness of correlations defined
in Eq.\ref{quant.part.re}.
\end{proof}


Theorem \ref{teo1} determines a new approach for the
quantification of
correlations between indistinguishable 
particles. 
Motivated by this result, it would be interesting to 
study the minimum of 
the single-mode correlations, described by the One-Way Work
Deficit in Eq.\eqref{defi2}, over all possible single-particle
representations. Recently,  Gigena and Rossignoli explored this idea
in entangled pure states with parity symmetry \cite{gigena15},  
and also investigated  the one-body information loss \cite{gigena16}.
 By means of Lemma \ref{lemma}, one can see that these two notions are indeed the 
one-way work deficit for fermions, as we describe below. 
A direct implication of Lemma \ref{lemma} is the analytical
computation of the One-Way Work Deficit for a single-particle mode: 
 \begin{equation}\label{quant.mode}
 Q^{j}_{a}(\rho) =
S \left( \ad{j} a_{j} \rho \ad{j} a_{j} + a_{j} \ad{j} \rho
a_{j} \ad{j} \right)- S(\rho),
 \end{equation}
where now we write  $Q^{j}_{a} \equiv Q^{(j)} $, 
in order to make clear that we deal with the quantumness of correlations of
a single-particle mode $a_j$. In particular, if the quantum state is
pure ($| \psi \rangle$), the One-Way Work Deficit is given by,
 \begin{equation}\label{av}
Q^{j}_{a}(\ket{\psi}) = H \left( \langle \ad{j} a_{j}
\rangle, \langle a_{j} \ad{j} \rangle \right).
 \end{equation}
Summing   the single-particle correlation of all modes,
$\sum_{j=1}^L Q^{j}_{a}({\rho})$,
 gives us  the  correlation in this particular basis of modes,
with minimum correlation in its single-particle modes, defining
in this way the {\it One-Body Quantumness of Correlations}:
\begin{equation}\label{quant.part}
\mathcal{Q}_{sp}(\rho) = \min\limits_{\{a\}}\left(
\sum_{j=1}^L Q^{j}_{a}(\rho) \right).
\end{equation}
As the quantifier 
in Eq.\eqref{defi2}, the One-Body Quantumness of Correlations  
is obtained by means of a quantifier of quantumness of
correlations
based on the tensor product construction of distinguishable
systems.
This is of paramount  importance, for  it circumvents any controversy 
about correlations
in Fock space. 
As discussed above, for single-particle modes partitioning there
exists only one projector that minimizes the local disturbance for
symmetric states,
then it is simple to show the equivalence of the quantumness of
correlations defined via one-way work deficit in
Eq.\eqref{defi2}, 
and the entanglement measure proposed in Ref.\cite{gigena15} for
pure states.
\begin{theorem} \label{teo2}
For a  pure state $\ket{\psi}\in\M{H}^2\otimes\M{H}^{2(L-1)}$,
the one-body quantumness of correlations and the {one-body 
entanglement} are given by:
\begin{equation}\label{gigena.theorem}
S_{sp}(\ket{\psi}) = \mathcal{Q}_{sp}(\ket{\psi}) =
\min\limits_{\{a\}}\left( \sum_{j=1}^L
Q^{j}_{a}(\ket{\psi}) \right),
\end{equation}
where $\M{E}_{sp}$ is the entropy of entanglement for pure states 
of particles, as proposed in Ref.\cite{gigena15}.
\end{theorem}
\begin{proof}
Considering the optimal projective measurement described by the map $\Pi_{j}^{a}$ with 
projectors $\{a_ja_j^{\dagger},a_j^{\dagger}a_j\}$, the quantumness of correlations in this mode representation is:
\[\mathcal{{Q}}_{sp}^a(\ket{\psi}) \equiv \sum_j Q_a^j(\ket{\psi}) = \sum_j
S(\Pi^a_{j}(\ket{\psi})),\]
where  $\Pi_{j}^{a}(\psi)= \langle\ad{j}a_j\rangle \ketbra{\phi_o}{\phi_o} +
\langle a_j \ad{j}
\rangle|{\phi_e}\rangle \langle
{\phi_e}|$ is a projective measurement described in
Eq.(\ref{proj.meas}). As this measurement is composed by rank-1
projectors, the states $\ket{\phi_o}$ and $\ket{\phi_e}$ are
orthogonal, therefore:
\begin{align}
\mathcal{{Q}}_{sp}^a(\ket{\psi}) = \sum_j H(\langle\ad{j}a_j\rangle,\langle a_j\ad{j}\rangle).
\end{align}
Taking the minimization over all isomorphisms $a: \M{F}_L
\leftrightarrow \M{H}^{2\otimes L}$, we obtain
\[ \min_{a} \mathcal{{Q}}_{sp}^a(\ket{\psi}) =\min_{a} \sum_j H(\langle\ad{j}a_j\rangle,\langle a_j\ad{j}\rangle)=S_{sp}(\ket{\psi}), \]
where $S_{sp}$ is the von Neumann entropy of the single
particles, as discussed in Ref.\cite{gigena15}, and quantifies
the entanglement and quantumness of correlations
of the single particles.
\end{proof}

We defined previously two different quantifiers for the
quantumness
of correlations, namely, $\mathcal{Q}_{sp}(\rho)$ 
and $\M{Q}^{\emptyset}_p(\rho)$. It is now important to show the 
interplay between these two quantifiers. Indeed this can be 
obtained using simple tools of  the quantum information formalism.

\begin{theorem}
For a state $\rho\in\M{D}({\M{F}_L})$, the zero-way work
deficit for
identical particles $\M{Q}_p^{\emptyset}$ is upper bounded by the 
one-body quantumness of correlations $\M{Q}_{sp}(\rho)$:
\begin{equation}
\mathcal{Q}_p^{\emptyset}(\rho) \leq \mathcal{Q}_{sp}(\rho).
\end{equation}
\end{theorem}
\begin{proof}
Consider the density matrix $\rho$,  represented in the single
particle basis
$\{ a_j^{k_j} \}_{j=1}^{L}$: 
\begin{equation}
\rho = \sum_{\vec{k},\vec{k'}} \omega_{\vec{k},\vec{k'}}
a_1^{k_1\dagger}\cdots a_L^{k_L\dagger}\ketbra{vac}{vac}
a_1^{k'_1}\cdots a_L^{k'_L}, 
\end{equation}
where $\vec{k} = (k_1, \ldots , k_L)$.
We can define the set of projective measurements
$\{\hat\Pi_{\vec{k}}\}$ such that:
\begin{align*}
\hat\Pi_{\vec{k}} &= a_1^{k_1\dagger}\cdots
a_L^{k_L\dagger}\ket{vac}\bra{vac}
 a_L^{k_L}\cdots a_1^{k_1}\\
 & = \ketbra{a_{\vec{k}}}{a_{\vec{k}}}, 
\end{align*}
where $\ket{a_{\vec{k}}}=a_1^{k_1\dagger}\cdots
a_L^{k_L\dagger}\ket{vac}$.
Now perform the local dephasing on the state:
\begin{align}
\Pi(\rho) &= \sum_{\vec{l}} \hat\Pi_{\vec{l}}\rho \hat\Pi_{\vec{l}}\\
& = \sum_{\vec{l}}
\Tr\left(\ketbra{a_{\vec{\ell}}}{a_{\vec{\ell}}}{\rho}\right)
\ketbra{a_{\vec{\ell}}}{a_{\vec{\ell}}}.
\end{align}
Let $p(\vec{l}) =
\Tr\left(\ketbra{a_{\vec{\ell}}}{a_{\vec{\ell}}}{\rho}\right)$ be
the
probability to find the modes occupied in $\vec{l} =
(l_1,l_2,\dots,l_L)$
configuration, for $l_j=\{0,1\}$.  The relative
entropy of $\rho$ and $\Pi(\rho)$ is:
\begin{align}
S(\rho||\Pi(\rho))& = S(\Pi(\rho)) - S(\rho) \\ \label{H}
& = H\{p(\vec{l})\}- S(\rho) ,
\end{align} 
where $H\{p(\vec{l})\}$ is the Shannon joint entropy of the
probabilities
$ p(\vec{l}) = p(l_1,l_2,...,l_L)$. 

Now let us obtain a relation between joint entropy and the
Shannon entropy. With the chain rule of the joint entropy for
$n$ random variables
$X_1,\dots,X_n$:
\begin{equation}
H(X_1,\dots,X_n) = \sum_{i=1}^n H(X_i|X_1,\dots,X_{i-1}), 
\end{equation}
and  the positivity of the conditional mutual information:
\begin{equation}
I(X_i|X_1,\dots,X_{i_1}) = H(X_i) - H(X_i|X_1,\dots,X_{i_1})\geq
0,
\end{equation}
we have:
\begin{equation}\label{inq.jent.}
H(X_1,\dots,X_n) \leq \sum_{i=1}^n H(X_i). 
\end{equation}
Then from  Eq.(\ref{H}) we obtain: 
\begin{align}
S(\rho||\Pi(\rho)) &\leq \sum_{j=1}^L H(p_j) -  S(\rho)\\
& \leq \sum_{j=1}^L  S^{(j)}(\rho||\Pi(\rho)),
\end{align}
where $H\{p_{j}\} = - p_0^j\log p_0^j- p_1^j\log p_1^j $, and $p_1^j(p_0^j)$ is  the probability to find a particle(hole) in
mode $j$.
In the second inequality above, we used that: $ S(\hat \Pi^{(j)}(\rho)) =
H\{p_{j}\} + \sum_l p^{j}_l S(\sigma_l^{L-j})\geq H\{p_{j}\}
$, 
in conjunction with Eq.(\ref{quant.mode}).  
Therefore,  as the last inequality holds for any projective
measurement, it
must also  hold for the optimal one, which proves the statement. 
\end{proof}

This result makes evident the nature of the construction of
these two quantifiers for the quantumness
of correlations. 
$\mathcal{Q}_{sp}(\rho)$ is based on the single particle mode
quantumness of correlations, 
 defined in Eq.\eqref{quant.mode}, which takes into account
the average
of binary entropies corresponding to the occupation of
particles/holes in
 each mode. On the other hand,  $\M{Q}^{\emptyset}_p(\rho)$  
is related to the joint probability for the occupation of
particles/holes in each mode.



\section{Activation Protocol}\label{act.prot.sec}
In this section, we characterize the class of fermionic states without quantumness of correlations,   by means of an activation protocol for a system of L modes. 
We show that the two notions discussed in the previous section (Eq.\eqref{eq:quantumness.particles} and
Eq.\eqref{quant.part})
 share the same set of states without quantumness of correlations.

A measurement process can be described by a unitary
interaction between the measurement
apparatus and the quantum system, followed by a projective
measurement on the apparatus.
Considering a system in  the state $\rho_{{S}}=\sum_k
\lambda_k \ketbra{k}{k}
\in\M{D}({H}_{{S}})$, the global initial state   
for  system/measurement-apparatus can be written as 
$\rho_{{S}:{M}} = \rho_{{S}}\otimes \ketbra{0}{0}_{{M}}$. The 
interaction between the system and the apparatus ancillary state
will
be performed by a unitary evolution:  
$U_{{S}:{M}}\in\M{U}({H}_{{S}}\otimes\mathcal{H}_{{M}})$, such
that
$Tr_{{M}}[U_{{S}:{M}} \rho_{{S}:{M}} U_{{S}:{M}}^{\dagger}]
= \sum_l\hat \Pi_l \rho_{{S}} \hat\Pi_l^{\dagger}$. 
A unitary operation satisfying this condition is given by:
\begin{equation}
U_{{S}:{M}}\ket{k}_{{S}}\ket{0}_{{M}} =
\ket{k}_{{S}}\ket{k}_{{M}},
\end{equation}
where $\{ \ket{k}\}$ is an orthonormal basis in
$\mathcal{H}_{{S}}$.
If the orthogonal basis $\{\ketbra{k}{k}\}$ is the canonical
one, this
interaction is the Cnot gate. Although this kind of interaction
creates only classical correlations for a global measurement process,
local measurements can create
entanglement between system and measurement apparatus
\cite{pianiprl,streltsov11}. The quantumness of correlations of
the system
can be obtained by means of the minimum amount of entanglement
created by the interaction:
\begin{equation}\label{act.prot}
Q_{E}(\rho_S) = \min_{V_{S}} E(\tilde{\rho}_{S:M}),
\end{equation}
where $V_S$ is a unitary, which sets the measurement basis.
  $\tilde{\rho}_{S:M} =
U_{S:M}(V_S\otimes\id_M)(\rho_S\otimes\ketbra{0}{0}_M)
(V_S^{\dagger}\otimes\id_M)U_{S:M}^{\dagger}$
is the result of the interaction. 
For each entanglement monotone $E$, it results in a different
quantifier of quantumness
of correlations $Q_{E}(\rho_S)$
\cite{pianiprl,streltsov11,taka,iemini14}.

\begin{figure}
\centering 
\includegraphics[scale=0.3]{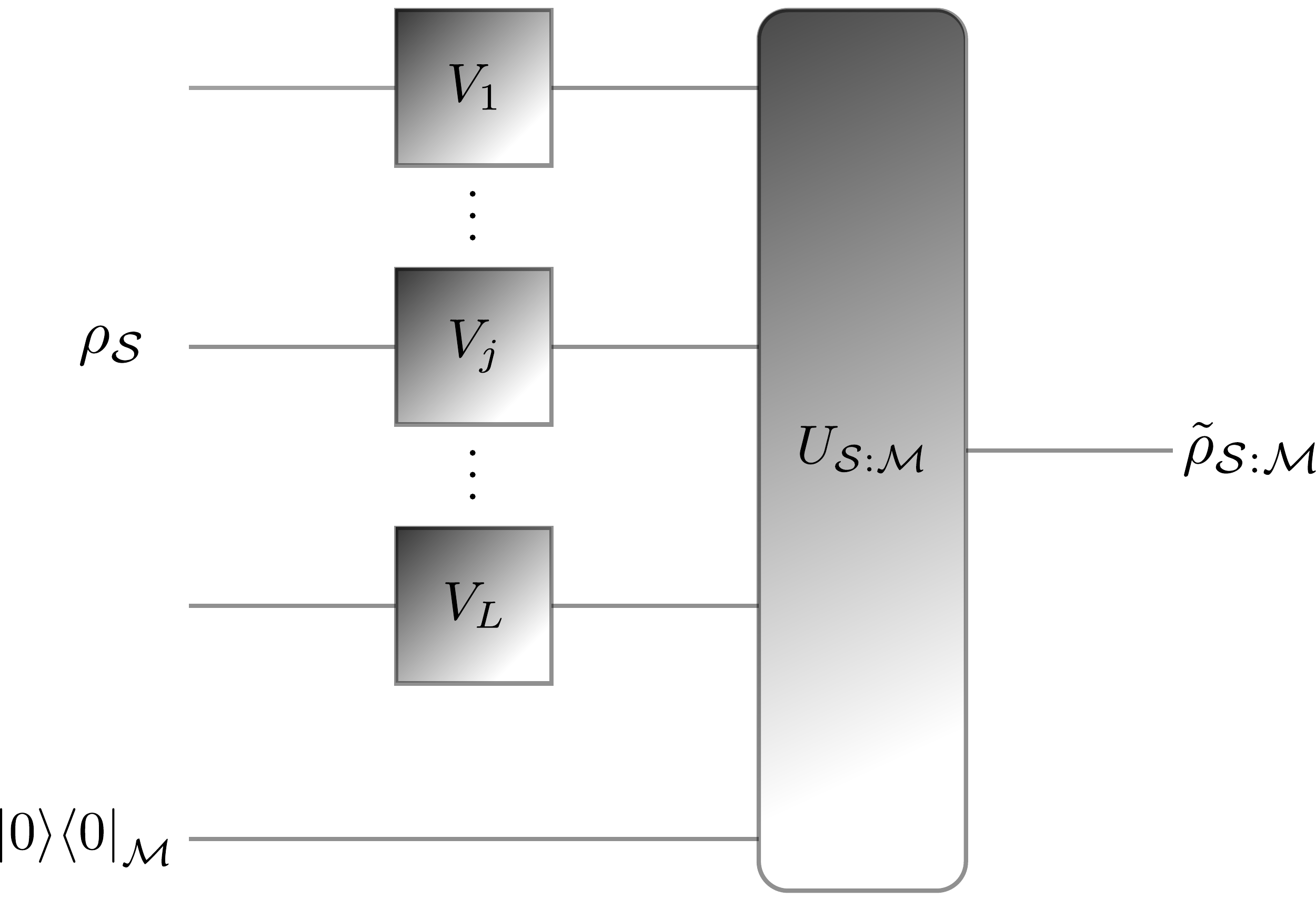} 
\caption{(Color online) Activation Protocol for an L-modes system.}\label{fig1}
\end{figure}

Fig.\ref{fig1} pictures the activation protocol for an L-modes 
system,  described by the density matrix
$\rho_S\in\M{D}(\M{H}_1^2\otimes\cdots\otimes\M{H}_L^2)$, where
the unitary operation $V_j$
represents rotation on mode $j$, and $U_{S:M}$ is a Cnot
operation with the
system as control and apparatus as target. The output state in
the quantum circuit is:
\[\tilde{\rho}_{S:M}=U_{S:M}(V_{S}\otimes\id_M)(\rho_S\otimes\ketbra{0}{0}_M)
(V_{S}^{\dagger}\otimes\id_M)U_{S:M}^{\dagger}, \]
where $V_{S} = V_1\otimes\cdots\otimes V_L$. The post
measurement state, resulted from the interaction between system
and
measurement apparatus is: 
\begin{equation}
\tilde{\rho}_{S} = \sum_{\vec{l}} \hat\Pi_{\vec{l}} V_S \rho_S
V_S^{\dagger}\hat\Pi_{\vec{l}},
\end{equation}
with $\hat\Pi_{\vec{l}}=\hat\Pi_{l_1}\otimes\cdots\otimes\hat\Pi_{l_L}$. As
aforementioned, 
the action of the unitary operation $V_S$ determines  the  projective measurement basis.
Therefore, from
Lemma \ref{lemma}, the unitary 
operation on each mode $j$ must preserve the symmetries of the
state in order to
minimize the local disturbance. Thus,  for one-body systems with parity symmetry,  
there exists
only one unitary operation $V_j=\id_j$ on each mode $j$. Therefore, 
the post measurement state, for one-body systems, is simply: 
\begin{equation}
\tilde{\rho}_{S} = \sum_{\vec{l}}\hat \Pi_{\vec{l}} \rho_S
\hat\Pi_{\vec{l}},
\end{equation}
with  $\hat\Pi_{l_j} = \{a_j^{\dagger}a_j,a_j a_j^{\dagger} \}$ in a
given
representation of modes.

We  will  use the approach of activation protocol to discuss
 the
classically correlated states.  As shown 
in Ref.\cite{pianiprl}, a quantum state $\chi$ is classically 
correlated if and only if there exists a unitary $V_S$ such
that:
\begin{equation}
\chi_S = \tilde{\chi}_{S} = \sum_{\vec{l}} \hat\Pi_{\vec{l}} V_S
\chi_S V_S^{\dagger}\hat\Pi_{\vec{l}},
\end{equation}
which immediately holds for the activation protocol of
$L$-modes
represented in Fig.\ref{fig1}.
Actually, we can learn about the set of classically correlated
states according to the 
one-body quantumness of particles,
by means of the representation of
modes.
As the set of classically correlated states of particles
satisfies $\M{Q}_{sp}(\chi)=0$, for all $\chi$ in this set,
there must exist an isomorphism $\Lambda: \M{F}_L \leftrightarrow
\M{H}^{2\otimes L}$ of one-body particles such that $\sum_j
\mathcal{Q}_{sp}^{(j)}(\chi)=0$.
This equality holds if and only if:
\[\mathcal{Q}_{sp}^{(j)}(\chi)=0, \qquad \forall \, j=1,\cdots,N. \]
Therefore, for all $j$, there is a projective measurement $\Pi^{(j)}$
such that $\Pi^{(j)}(\chi) = \chi$. Thus a quantum state
$\chi\in\M{D}(\M{F}_L)$ is classically correlated,
if and only if, there exist projective measurements in each mode
$\Pi^{(1)}\otimes\cdots\otimes\Pi^{(L)}(
\chi) = \chi$, such that
\[ \chi = \sum_{\vec{l}} p(\vec{l})
\hat\Pi_{l_1}\otimes\cdots\otimes\hat\Pi_{l_L}, \]
where $\vec{l} = (l_1,\ldots,l_L)$ and $l_j = 0,1$.
 From the
definition of quantumness of correlations of particles in
Eq.(\ref{quant.part}), and the
isomorphism in Eq.(\ref{isomorphism}), we can write the
projectors as $\hat\Pi_{\vec{l}} = a_1^{l_1\dagger}\cdots
a_L^{l_L\dagger}\ket{vac}\bra{vac}a_L^{l_L}\cdots a_1^{l_1}$,
and  $\hat\Pi_{\vec{l}}=\hat\Pi_{l_1}\otimes\cdots\otimes\hat\Pi_{l_L}$,
therefore the classically correlated states of particles can be
written as:
\begin{equation}
\chi = \sum_{\vec{l}} p(\vec{l}) a_1^{l_1\dagger}\cdots
a_L^{l_L\dagger}\ket{vac}\bra{vac}a_L^{l_L}\cdots a_1^{l_1}.
\end{equation}
This results is in agreement with the set of classically
correlated states obtained in Ref.\cite{iemini14}, by means of
the activation protocol \cite{pianiprl,streltsov11}, for the zero
way work deficit
introduced in Eq.(\ref{quant.part.re}). As a matter of fact,  as the set of
classically correlated states of particles is independent of the isomorphism,
we can write:
 \begin{equation}
 \M{Q}_p^{\emptyset}(\chi) = \mathcal{Q}_{sp}(\chi)=0.
 \end{equation}
It is important to note that the classicality of correlations of
particles, for a given state, does not imply in a classically
correlated state for a given mode representation. In other
words, the system of particles can be classically correlated,
whereas there exist modes that are quantum correlated. The requirement,  
for the particles be classically correlated,
in Eq.\eqref{quant.part}
is the existence of at least 
one representation $\Lambda: \M{F}_L \leftrightarrow
\M{H}^{2\otimes L}$ such that the modes are classically
correlated \cite{zanardprl}.

\section{Dissipative System}\label{dissip.system}

In this section,  we study a physical system and its quantumness of correlations
 in order to illustrate the previous discussions. We investigate 
a purely dissipative system, whose stationary states exhibit
interesting topological non-local correlations.
 The motivation for the study of a dissipative system stems from  the fact
 that its dynamics tends to mix the quantum state, allowing 
 us to study the quantumness of correlations beyond entanglement, since for
 pure states such notions usually overlap.
 Furthermore, we consider a system which conserves the total number of particles, 
 as we will describe in more detail later, and in this way we can
 use our previous results for the quantumness of symmetric states.    
 Let us give a brief overview of  dynamics in open systems and
  present the physical setting under scrutiny.

In general, evolutions in dissipative open quantum systems 
tend to annihilate the quantum correlations present in the 
system, leading to steady states described by trivial
mixed states.
 There are cases, however, in which the
many-body density matrix is driven towards a pure
steady state, $\rho_f = \ketbra{\psi_D}{\psi_D}$, commonly
called as dark
states in quantum optics.
 Recently, theoretical and experimental studies 
 \cite{Iemini_Dissipation_2016,diehl08,verstraete09,diehl11}
 have focused on how to properly engineer the environment such
that, in the long-time limit, 
 it drives the system into
certain desired quantum states. 
In particular, we study here the dissipative number conserving 
system as proposed in \cite{Iemini_Dissipation_2016},
 consisting of a single, or two-leg ladder, 
 properly coupled to the environment. It was shown
that such setup leads to topological superconductors as steady
states,
 exhibiting non-local edge correlation and Majorana zero modes,
  with promising applications for topologically protected 
  quantum memory and computing \cite{review.Nayak_2008}.

Let us formally describe the system we will study. The
time
evolution of a system coupled to
 a Markovian reservoir (memoryless reservoir) is described by
 a master equation, cast in the following form,
\begin{eqnarray}
\frac{\partial \rho}{\partial t} = \hat{\mathcal{L}}[\rho] =
-i\, \comm{\hat H}{\rho} + J \sum_j \left( \hat L_j \rho \hat
L^{\dagger}_j -
\frac{1}{2}\anti{\hat L^{\dagger}_j \hat L_j}{\rho}\right),
\nonumber \\
\label{master.equation}
\end{eqnarray}
where $\rho$ is the density matrix, $\hat{\mathcal{L}}$ is the
Liouvillian of the evolution, $\hat H$ is the Hamiltonian of the
system,
and $\hat L_j$ are the Lindblad operators. Considering a purely dissipative evolution ($\hat H=0$), possible pure dark states in
the
 system are mathematically
 related to zero modes of the Liouville operator. More
precisely, dark states are zero modes shared by all Lindblad
operators:
\begin{equation}
\hat L_j  \ket{\psi_D} = 0, \quad \forall j.
\label{dark.state.def}
\end{equation}
Their existence, however, is not always guaranteed. 
 We study a one-dimensional 
fermionic system with $L$ sites, evolved by the following
number-conserving
 Lindblad operators:
\begin{equation}
\hat L_j =  (\ad{j} + \ad{j+1}) (a_j - a_{j+1}),
\end{equation}
where $a^{(\dagger)}_j$ are 
the fermionic annihilation (creation) operators at the site $j$,
and we consider
open boundary conditions ($j=1,...,L-1$).
In this setting,  the dark states are $p$-wave superconductors
 with fixed number of particles.
 
Let us focus now in the simpler non-trivial fermionic system which can
 present quantum correlations beyond the mere exchange statistics,
 \textit{i.e.,} a system with 
 $L=4$ sites and $N=2$ particles. Even for such a small setting,
 the dissipative system is already able to create superconducting correlations between the
  particles. Interestingly, the exact expression for such steady state in real space 
  was studied not only in the realm of dissipative systems 
   \cite{Iemini_Dissipation_2016}, but also as ground states of a closed Hamiltonian \cite{iemini_Maj_H}, 
 and is given by
   the equal weighted superposition
  of all possible configurations of its $N$ particles in the $L$ sites; precisely,
\begin{eqnarray*}
\ket{\psi_D} = \frac{\left(\ad{1}\ad{2} + 
\ad{1}\ad{3} + \ad{1}\ad{4} + 
 \ad{2}\ad{3} + \ad{2}\ad{4} + \ad{3}\ad{4} \right)}{\sqrt{6}} \ket{vac}. \nonumber \\
\end{eqnarray*} 

Since the above setting conserves the total number of particles, we can use our results for the quantumness
 in symmetric states.
Thus, we study the dynamics for an initial uncorrelated
single-Slater determinant
state,
 \begin{equation}
 \ket{\psi(t=0)} = \ad{1} \ad{3} \ket{vac}.
 \label{eq:intial.state.diss}
 \end{equation}

In order to characterize the time evolution described by the
master
equation (Eq.\eqref{master.equation}), we use the Runge-Kutta
integration.
 This method entails an error due
to inaccuracies in the numerical integration, but the full
density
matrix is represented without any approximation.
\begin{figure}
\includegraphics[scale=0.45]{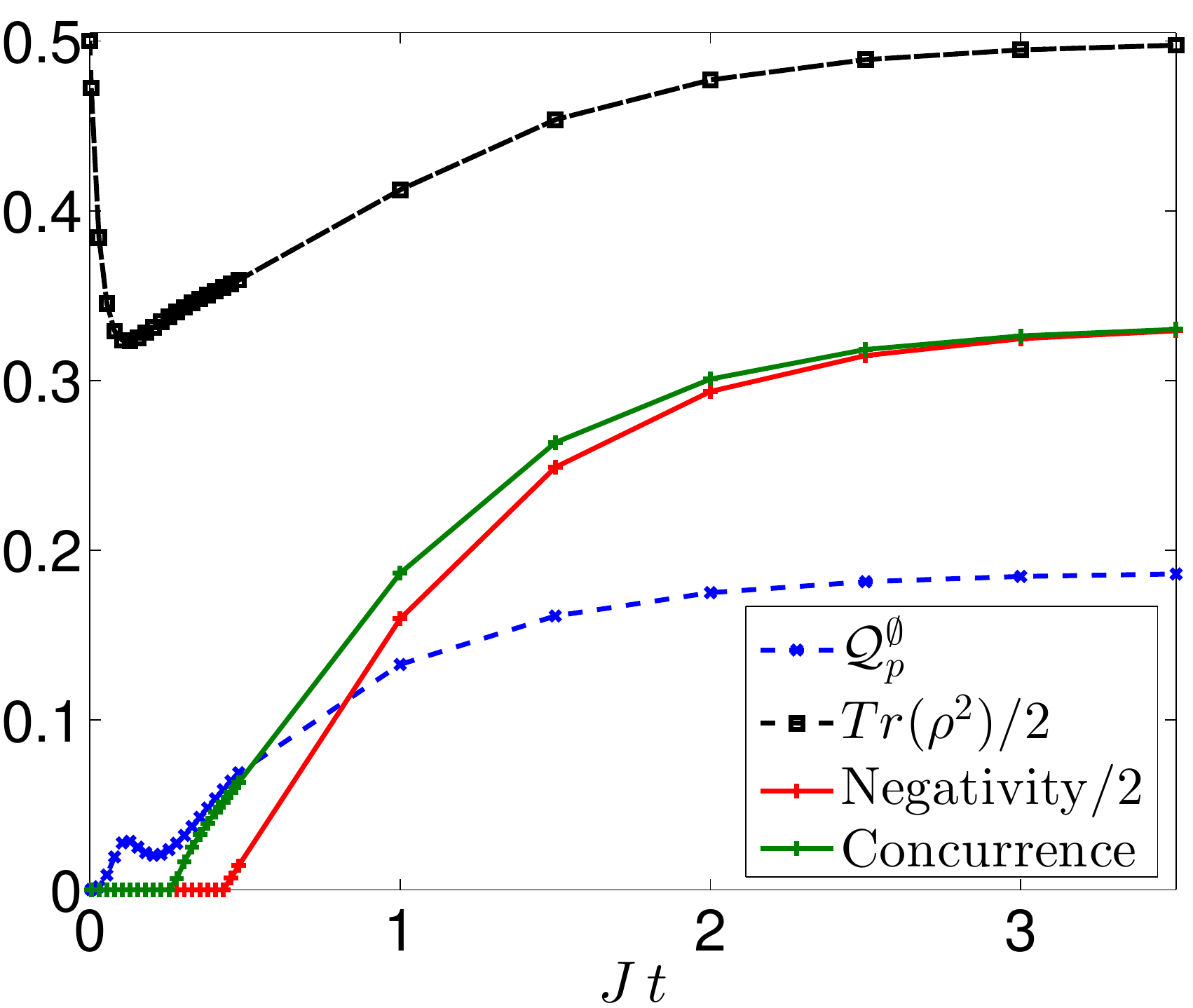}
\caption{ (Color online) Dissipative evolution for a system with $L=4$ sites and $N=2$ particles, 
considering an initially uncorrelated single-Slater 
determinant state (Eq.\eqref{eq:intial.state.diss}). We show the time evolution
for the quantumness of correlation of the state ($\mathcal{Q}_p^{\emptyset}$), its purity ($\Tr(\rho^2)$) and its entanglement according to the Shifted Negativity \protect\cite{iemini_ComputableMeasures} and Concurrence \protect\cite{eckert02} quantifiers.}
\label{fig:diss.evolution}
\end{figure}

We present our numerical simulation in
Fig.\eqref{fig:diss.evolution}.  At the early stages of  the
evolution, 
the quantum state becomes highly mixed,  and
we already see the creation of quantumness between its
particles.
Already in the beginning of the evolution,  $tJ \ll 1$, several excited modes 
 of the Liouvillian have non-trivial effects in the dynamics,
and in this way the observables present fast oscillations in
this regime.
For longer times,  only the first excited states of the
Liouvillian
 have non-trivial effects in the evolution,
 and the dynamical behavior becomes smoother. We see that
 the quantum state 
 is driven, exponentially fast in time, to a pure state with 
 non-trivial quantumness of correlations,
 as expected from the previous discussions.
Interesting to notice that the quantumness for the steady state agrees with the entanglement entropy for indistinguishable 
 particles \cite{iemini_ComputableMeasures},
 \textit{i.e.,} 
$\mathcal{Q}^{\emptyset}_p(\ketbra{\psi}{\psi}_D) =
S(\ketbra{\psi}{\psi}_D) - \ln N$.

{ To exemplify the different nature of the quantumness
 of correlations to the entanglement of particles, we also plot 
 the entanglement dynamics according to some usual 
 quantifiers in the literature. Precisely, we 
 compare the quantumness of correlations to the 
 Shifted Negativity \cite{iemini_ComputableMeasures} 
  and Concurrence \cite{eckert02}.
 It it very clear the more general character 
 of the quantumness of correlation,
 showing a richer dynamics for the initial evolution, while 
 there is absolutely no entanglement in the state. We can also notice that the state 
 has bound entanglement, analogous to the positive partial transpose (PPT)
  entangled states in distinguishable systems, for values of $Jt\in[0.3,0.45]$, as 
 the negativity is null, besides the concurrence has non zero values. }

\section{Local Projectors Structure - Numerical Results}\label{num.res}

 In this section we analyze, numerically, the quantumness of more general 
 symmetric quantum states 
 and their respective local projectors, as given in Definition 1. 
  Our motivations here are two-fold: (i) to analyze if there are 
  other solutions for the local projectors in Lemma 1, 
  \textit{i.e.,} solutions which do not share the quantum state symmetry; 
 (ii) to  analyze if Lemma 1 could be extended for  
  general states with parity symmetry, $[\rho, (-1)^{\hat N}]=0$, beyond  
  the restriction to a single 
  symmetry eigenvalue.
\begin{figure*}
\includegraphics[scale=0.3]{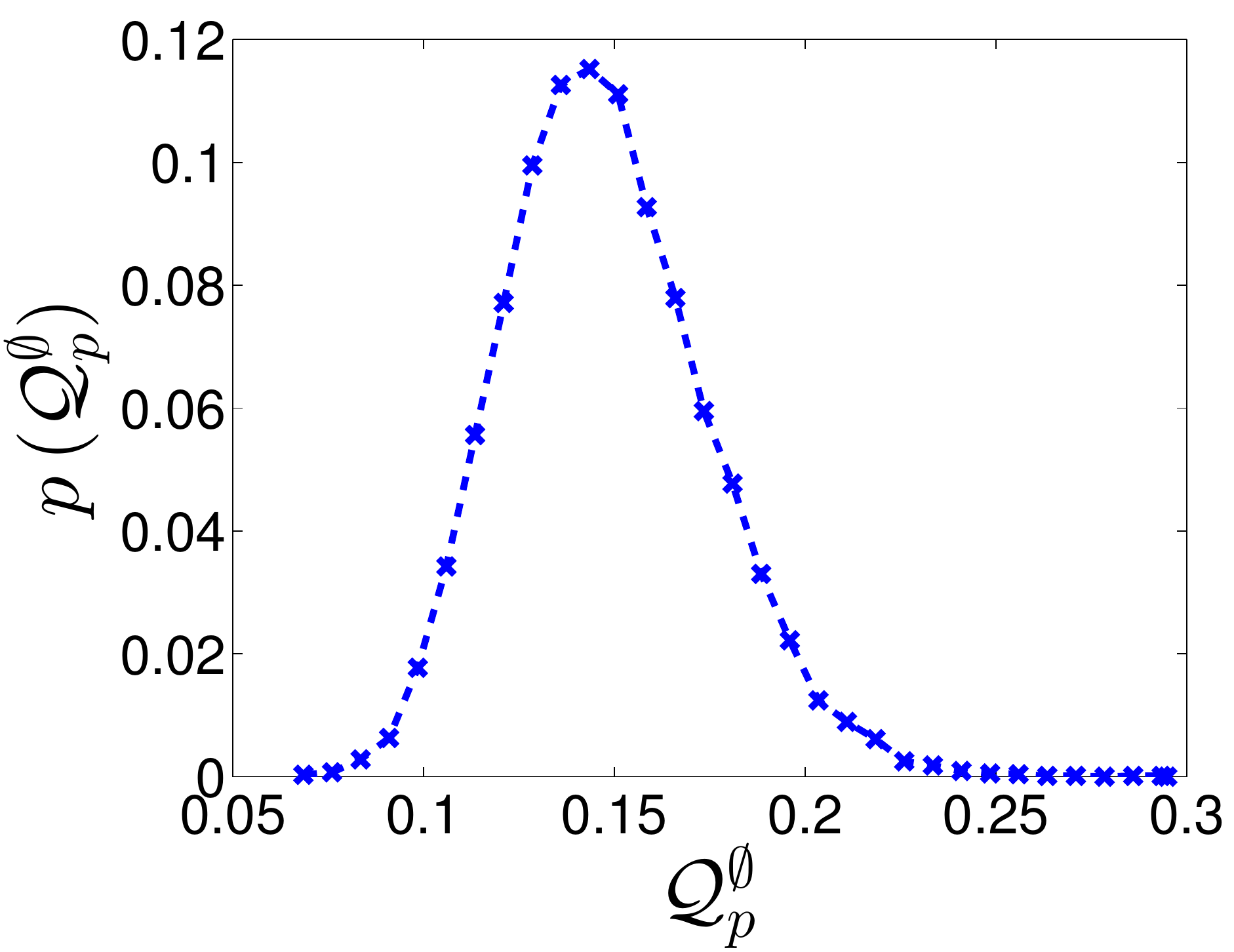}
\includegraphics[scale=0.3]{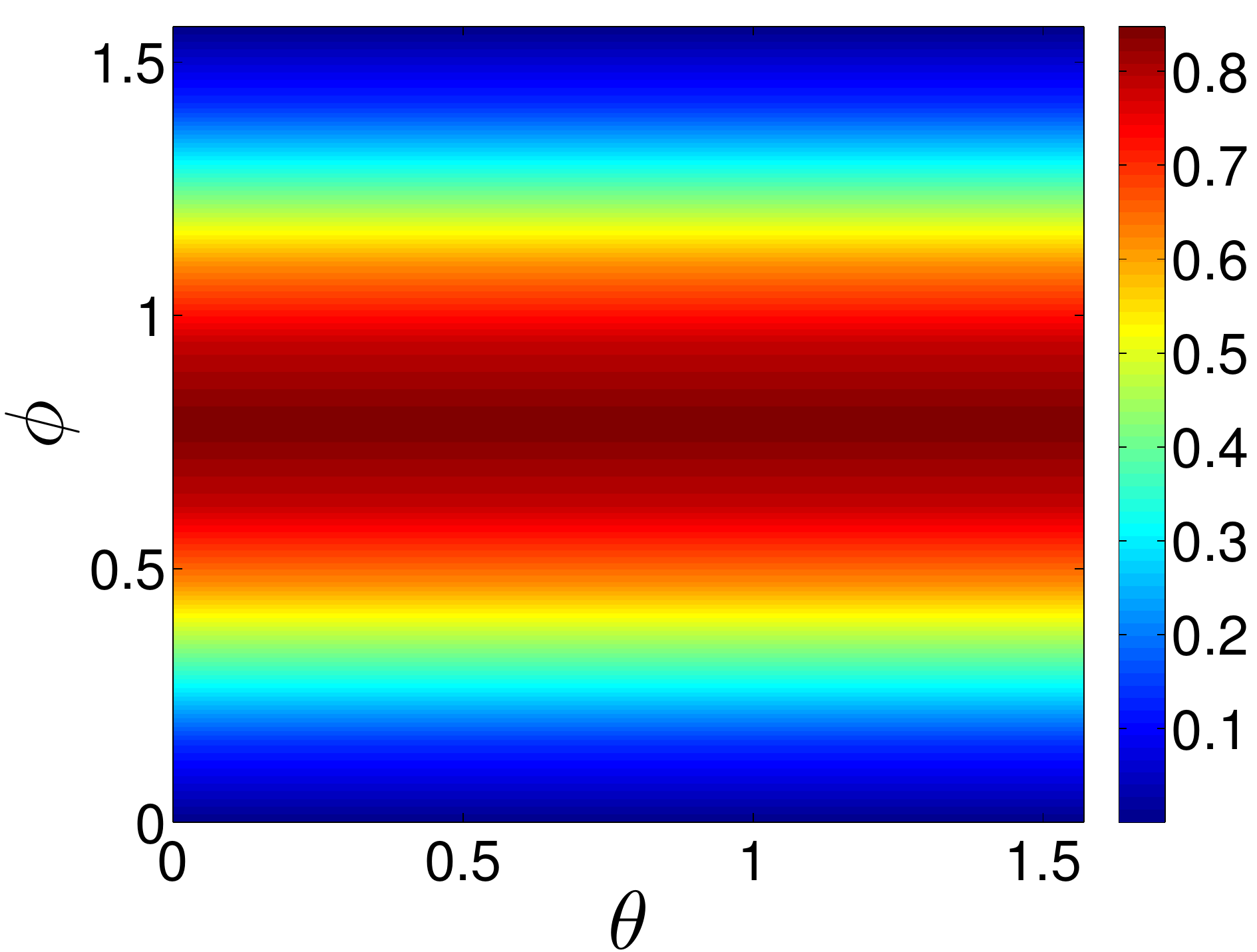}
\includegraphics[scale=0.3]{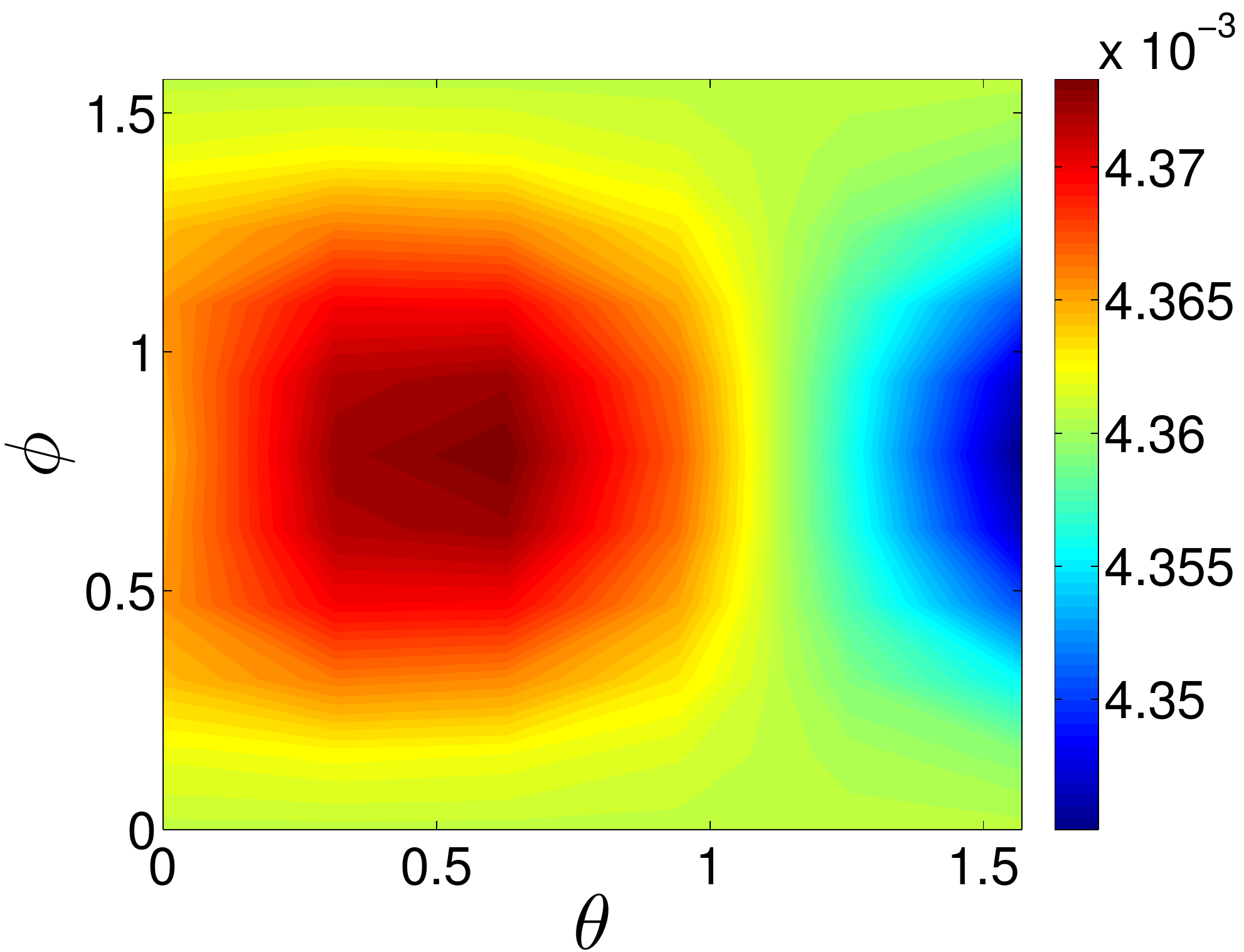}
\caption{(Color online)  On the left, we show the probability distribution for 
the quantumness of correlations in the sampled space 
to $E_{par}$. We use 
a sample space of $\sim 10^5$ quantum states. In the middle (on the right)
we show  the $T_E(\phi,\theta)$	 function for the $E_{par}^{(1)}$ ($E_{par}$) ensemble.} 
\label{fig:num.quantumness}
\end{figure*}

 We focus on a simple bipartite case, 
formed by  a  single particle
 mode ($1$ qubit) and the other  $L-1$ modes.
In this case, we can parametrize
 the local projectors in the single particle mode subspace with  
only two parameters, $\{\phi,\theta\}$. Since the projectors are rank-$1$, we need to  parametrize  only two orthogonal pure
states
 in such subspace,  as follows,
 \begin{eqnarray}
\ket{\psi_1(\phi,\theta)} &=& \cos(\phi) \ket{0} + e^{i
\theta}\sin(\phi)\ket{1}\\
\ket{\psi_2(\phi,\theta)} &=& -e^{-i\theta}\sin(\phi)\ket{0} +
\cos(\phi)\ket{1};
 \end{eqnarray}
where $\ket{1}$ ($\ket{0}$) is the state with one (no) fermion occupying the single particle mode, $0 \leq \phi \leq \pi$ and $ 0\leq \theta < 2\pi$. The local projectors are thus 
defined as $ \Pi_{1(2)}^{(\phi,\theta)} = 
\ketbra{\psi_{1(2)}(\phi,\theta)}{\psi_{1(2)}(\phi,\theta)}$. 
We define now the function $T_E(\phi,\theta)$ for such
projectors as,
\begin{equation}
T_E(\phi,\theta) = \int_{\rho \in E} \left\{ S(\rho ||
{\Pi}^{(\phi,\theta)}[\rho]) -\mathcal{Q}_p^{\emptyset}(\rho) \right\}
d\rho,
\label{eq:T.projectors}
\end{equation}
where ${\Pi}^{(\phi,\theta)}[\rho] = 
\hat{\Pi}_1^{\phi,\theta}\rho \hat{\Pi}_1^{\phi,\theta} + 
\hat{\Pi}_2^{\phi,\theta}\rho \hat{\Pi}_2^{\phi,\theta}$, and
the integral
is taken over an ensemble $E$ of quantum states. 
The above function gives us a notion of an effective
perturbation - in comparison
to the minimum one -
of the local projectors onto the corresponding ensemble $E$.

We consider an ensemble $E_{par}$ with parity symmetry,
$ [\rho,(-1)^{\hat N}]=0,$ $\forall \rho \in E_{par}$, and a
subset $E_{par}^{(1)}$
thereof, whose quantum states have a single eigenvalue of the
parity symmetry.
 This latter case - $E_{par}^{(1)}$ -
 corresponds to the conditions of Lemma 1.

  The integration of Eq.\eqref{eq:T.projectors} is
 performed over a  sample 
of quantum states ($\sim 10^5$ states) approximating  the corresponding
ensemble. The sample
is chosen according to the Haar measure. On the left of
Fig.\eqref{fig:num.quantumness}, we plot
the probability distribution for the quantumness of correlations of the sampled space corresponding to the $E_{par}$ ensemble.
 
Let us now analyze the $T_E(\phi,\theta)$ function for the
$E_{par}^{(1)}$ and
$E_{par}$ ensembles. In the former case (middle of 
Fig.\eqref{fig:num.quantumness})
 we obtained
that $T(\phi,\theta) = 0$ $\iff$ $\phi = 0,\pi/2$ (we omit here
$\phi>\pi/2$ for
 symmetry reasons), 
 which corresponds to the projectors with the shared symmetry, 
 as expected from our Lemma.
It was also observed in our simulations that $S(\rho ||
{\Pi}^{(\phi,\theta)}[\rho])
 - \mathcal{Q}(\rho) = 0$ $\iff$ 
$\phi = 0,\pi/2$, confirming  that the only optimal  projectors  are indeed those of the Lemma. 
On the right of Fig.\eqref{fig:num.quantumness},  we show our results
for
 the more general ensemble $E_{par}$. We see
now that the Lemma cannot be extended for such ensemble. We also see that  there is no 
 set of symmetric local projectors minimizing the local disturbance.
The analysis of the quantumness of correlations in these cases demands a more
careful treatment.
 
\section{Conclusion}\label{conc}

 In this work we proposed a new description for the 
 quantification of quantumness of correlations in fermionic systems. 
We proved in Lemma~\ref{lemma} that the symmetries of a state 
can improve the optimization of the local disturbance, once that 
 the optimal local projective measurement is also symmetric. As discussed, 
 it holds for symmetric states with a single eigenvalue of the symmetry operator. 
 Numerical evidences suggest the uniqueness of the symmetric solution 
 for the minimal local disturbance, as well as the impossibility of extending 
 the Lemma~\ref{lemma} for states with multiple eigenvalues of the symmetry operator.  
In Theorem~\ref{teo1},  we restrict our discussion for states with parity symmetry, showing
that the minimization of the multipartite relative entropy of quantumness reduces to the notion of 
quantumness of correlations of indistinguishable particles. 
By means of the activation protocol, we have also characterized 
the class of fermionic states without quantumness of correlations.
We illustrated our results with  the dynamics of 
quantumness of correlations for a purely dissipative system
 of two particles and four sites. 
Our results  shed new light and  give  fresh perspectives on the 
characterization and quantification of quantum correlations in fermionic systems. 

\acknowledgments
This work is supported by INCT-Quantum Information, CNPq and FAPEMIG. We would like to thank F.F. Fanchini and P.H.S. Ribeiro by fruitfull discussions.  
 
\section*{Appendix -  Proof of Lemma 1}

\begin{proof}
Consider a symmetry $\Theta$ with $N$ degenerate eigenvalues
$\{\theta_j\}_{j=1}^N$. In its eigenbasis, it can be written as:
\begin{equation}
\Theta = \begin{pmatrix}
\hat{\theta}_1 & \cdots & 0 \\
0  & \ddots & 0\\
0 & \cdots & \hat{\theta}_N  
\end{pmatrix}, \end{equation}
where $\hat{\theta}_j$ are $n_j\times n_j$ matrices, with
$n_j$ being    the degeneracy of eigenvalue $\theta_j$. The observable
$\Theta$ acts
on a D-dimensional Hilbert space $\M{H}^D$, where $D=
\sum_{j=1}^N n_j$.

A given state $\rho_j$ has a symmetry over an eigenvalue
$\theta_j$ if $[\rho_j, \hat{\theta}_j]=0$.
Therefore, a system described by a density matrix
$\rho_S\in\M{D}(\M{H}_S)$ has the
symmetry $\Theta$, if the density matrix can be written as a
convex combination of the set of $\rho_j$ for all eigenvalues
$\theta_j$, namely:
\[ \rho_S = \sum_j q_j \rho_j.\]
Thus there exists an isometry $V_X:\M{H}_S\rightarrow
\M{H}_S\otimes\M{H}_X$,  which acts
on the eigenbasis of $\Theta$ as:
\[ V_X \ket{\theta_j(l)}_S = \ket{\theta_j(l)}_S\ket{\theta_j}_X
,\]
where $\{\ket{\theta_j(l)}_S\}_{l=1}^{n_j}$ are the eigenvectors
related to the eigenvalue $\theta_j$. 
 $\M{H}_S$ is
the Hilbert space of the system, and $\M{H}_X$ is an ancillary
space such that $\text{dim}(\M{H}_S)\times
\text{dim}(\M{H}_X)=D$.
Therefore, the action of the isometry over the symmetric density
matrix $\rho_S$ is:
\[ \rho_{SX}=V_X\rho_SV_X^{\dagger} = \sum_j q_j
\rho^S_j\otimes\ketbra{\theta_j}{\theta_j}_X, \]
which is a block diagonal matrix, with the blocks labeled by the
eigenvalues of $\Theta$, and consequently $[ \rho_{SX},\Theta]=0$.
The set $\{\ket{\theta_j}\}_{j=1}^N$ is 
an orthonormal basis on $\M{H}_X$. Each density matrix
$\rho_j^S$ acts on the space spanned by the eigenvectors of the
eigenvalue $\theta_j$.

We separate the projective measurements over a symmetric state
in two kinds: with and without the symmetry. To represent these
two different
measurements we use the approach of an enlarged state space, 
described below.
The way that the measurement acts over the space $\M{H}_X$
defines if it has or has not the symmetry.
A projective measurement with the symmetry must keep the block
diagonal form of the density matrix, respecting the label
on the eigenvalues of $\Theta$ described by the basis
$\{\ket{\theta_j}\}_j$ on $\M{H}_X$. On the other hand, a
projective measurement without
the symmetry will create overlap with this same basis,
destroying the symmetry and the block diagonal structure of the
density matrix in the
enlarged space. 
In the case of a modes partitioning, the Hilbert space of the
system is composed as
 $\M{H}_S = \M{H}_A\otimes
\M{H}_B$.
Thus to prove the Lemma, we define two local projective
measurement maps:
a $\M{P}^{\theta}_{BX}$ that has the symmetry on $\Theta$, and a
map $\Pi_{BX}$
without the symmetry. 
We obtain the proof of the Lemma by showing  that the local
disturbance, created on
$\rho_{ABX}$ by $\M{P}^{\theta}_{BX}$, is smaller than that
crated by $\Pi_{BX}$, for a state with symmetry over one
eigenvalue $\theta_j$:
\begin{equation}
S(\rho_{ABX}||\M{P}^{\theta}_{BX}(\rho_{ABX}) ) \leq
S(\rho_{ABX}|| \Pi_{BX}(\rho_{ABX}) ).
\end{equation}

As $\M{P}^{\theta}_{BX}$ has the symmetry over the eigenvalues
of $\Theta$,  it must act on $\rho_{ABX}
=\sum_j q_j \rho^{AB}_j\otimes\ketbra{\theta_j}{\theta_j}_X$
inside the blocks:
\begin{equation}\label{ptheta}
 \M{P}^{\theta}_{BX}(\rho_{ABX}) =  \begin{pmatrix}
q_1\M{P}_{B}(\rho^{AB}_1) & \cdots & 0 \\
0  & \ddots & 0\\
0 & \cdots & q_N\M{P}_{B}(\rho^{AB}_N)
\end{pmatrix},\end{equation}
where $\M{P}_{B}(\rho^{AB}_j)$ is a local projective measurement
over subsystem $B$. For the projective measurement to satisfy
Eq.\eqref{ptheta}, the projectors must be in the form:
\[P_{b,x} = P_b^B\otimes\ketbra{\theta_x}{\theta_x}_X. \]
Then its action over $\rho_{ABX}$ respects:
\[ \M{P}^{\theta}_{BX}(\rho_{ABX}) = \sum_j q_j
\B{P}_B(\rho^{AB}_j)\otimes{\theta_j}_X, \]
where $\B{P}_B(\rho^{AB}_j)$ is a local projective measurement
map over subsystem $B$, with projectors $\{P_b \}$,
as presented in Eq.\eqref{gen.projective.meas}. For the disturbance
created by this
measurement,  we can write:
\[ S(\rho_{ABX}||\M{P}^{\theta}_{BX}(\rho_{ABX}) ) = \sum_j q_j
S(\rho^{AB}_j||\B{P}_{B}(\rho^{AB}_j) ), \]
once that $S(\sum_x p_x \ketbra{x}{x}\otimes\rho_x||\sum_x p_x
\ketbra{x}{x}\otimes\sigma_x) = \sum_x p_x S(\rho_x||\sigma_x)$.
On the other hand, the projective measurement $\Pi_{BX}$ creates
an overlap on the orthonormal basis $\{\ket{\theta_j}\}$.
It must act over space $B$ and $X$ locally, without creating
correlations between them, however creating overlap
between the basis $\ket{\theta_j}$ and another basis in
$\M{H}_X$.
We have:
\begin{equation} \Pi_{b,x} = \Pi_b \otimes \ketbra{x}{x}_X,\end{equation}
for $\{ \ket{x}\}$ an orthonormal basis in $\M{H}_X$, thus the
action of the map can be written as:
\[ \Pi_{BX}(\rho_{ABX}) = \sum_j q_j \sum_x p(x|j)
{\Pi}_B(\rho_j^{AB})\otimes \ketbra{x}{x}_X, \]
where $p(x|j) = |\braket{x}{\theta_j}|^2$ represents the overlap
between the eigenstates of $\Theta$ under the action of the
projective measurement.
As the relative entropy decreases under the partial trace operation,
then tracing over subsystem $X$, the local disturbance created
by $\Pi_{BX}$ satisfies:
\[ S(\rho_{ABX}|| \Pi_{BX}(\rho_{ABX}) ) \geq S(\sum_j q_j
\rho_j^{AB} || \sum_j q_j{\Pi}_B(\rho_j^{AB})), \]
with $\sum_x p(x|j) = 1$ and $\Tr_X(\rho_{ABX}) = \sum_j q_j
\rho_j^{AB}$.

Therefore, considering a bipartite state $\rho_l^{AB}$ with
symmetry in only one eigenvalue $\theta_l$ of $\Theta$, 
there exists just one term in the
sum $\sum_j q_j \delta_{j,l}= q_l =1$, which implies $\rho_{ABX} =
\rho_l^{AB}\otimes\ketbra{\theta_l}{\theta_l}_X$, and the
disturbances satisfy:
\begin{equation}\label{res1}
S(\rho_{ABX}||\M{P}^{\theta}_{BX}(\rho_{ABX}) ) =
S(\rho^{AB}_l||\B{P}_{B}(\rho^{AB}_l) ),
\end{equation}
and 
 \begin{equation}\label{res2}
S(\rho_{ABX}|| \Pi_{BX}(\rho_{ABX}) ) \geq S(\rho_l^{AB}
||{\Pi}_B(\rho_l^{AB})).
\end{equation}
The optimization in the one-way work deficit is taken over all
projective measurements that act over subspace $B$. Therefore,
as ${\Pi}_B(\rho_l^{AB})$ and $\B{P}_{B}(\rho^{AB}_l)$ are
restricted by the symmetry to act on this same space, and 
the state has null projection on any other subspace of the symmetry,
the smallest local disturbance created by these two projective 
measurement can attain the same value:
\[\min_{\B{P}_{B}}S(\rho^{AB}_l||\B{P}_{B}(\rho^{AB}_l) )=
\min_{{\Pi}_{B}}S(\rho_l^{AB} || {\Pi}_B(\rho_l^{AB})). \]
Finally,by  Eq.\eqref{res1} and
Eq.\eqref{res2}, we obtain: 
\[S(\rho_{ABX}||\M{P}^{\theta}_{BX}(\rho_{ABX}) ) \leq
S(\rho_{ABX}|| \Pi_{BX}(\rho_{ABX}) ), \]
which means that local projective measurements, with the symmetry
of the state, create less disturbance than
local projective measurements without the symmetry, proving the
Lemma.

Besides the proof was performed for bipartite systems, it also
holds for multipartite systems, simply generalizing the
projective measurements over partition $B$ to  multipartite
projectors.
\end{proof}


\begin{thebibliography}
\expandafter\ifx\csname natexlab\endcsname\relax\def\natexlab#1{#1}\fi
\expandafter\ifx\csname bibnamefont\endcsname\relax
  \def\bibnamefont#1{#1}\fi
\expandafter\ifx\csname bibfnamefont\endcsname\relax
  \def\bibfnamefont#1{#1}\fi
\expandafter\ifx\csname citenamefont\endcsname\relax
  \def\citenamefont#1{#1}\fi
\expandafter\ifx\csname url\endcsname\relax
  \def\url#1{\texttt{#1}}\fi
\expandafter\ifx\csname urlprefix\endcsname\relax\def\urlprefix{URL }\fi
\providecommand{\bibinfo}[2]{#2}
\providecommand{\eprint}[2][]{\url{#2}}

\bibitem[{\citenamefont{Nielsen and Chuang}(2010)}]{chuang}
\bibinfo{author}{\bibfnamefont{M.~A.} \bibnamefont{Nielsen}} \bibnamefont{and}
  \bibinfo{author}{\bibfnamefont{I.~L.} \bibnamefont{Chuang}},
  \emph{\bibinfo{title}{{Quantum computation and quantum information}}}
  (\bibinfo{publisher}{Cambridge university press}, \bibinfo{year}{2010}).

\bibitem[{\citenamefont{Nayak et~al.}(2008)\citenamefont{Nayak, Simon, Stern,
  Freedman, and Sarma}}]{review.Nayak_2008}
\bibinfo{author}{\bibfnamefont{C.}~\bibnamefont{Nayak}},
  \bibinfo{author}{\bibfnamefont{S.~H.} \bibnamefont{Simon}},
  \bibinfo{author}{\bibfnamefont{A.}~\bibnamefont{Stern}},
  \bibinfo{author}{\bibfnamefont{M.}~\bibnamefont{Freedman}}, \bibnamefont{and}
  \bibinfo{author}{\bibfnamefont{S.~D.} \bibnamefont{Sarma}},
  \bibinfo{journal}{Rev. Mod. Phys.} \textbf{\bibinfo{volume}{80}},
  \bibinfo{pages}{1083} (\bibinfo{year}{2008}).

\bibitem[{\citenamefont{Zanardi}(2002)}]{zanardi02}
\bibinfo{author}{\bibfnamefont{P.}~\bibnamefont{Zanardi}},
  \bibinfo{journal}{Physical Review A} \textbf{\bibinfo{volume}{65}},
  \bibinfo{pages}{042101} (\bibinfo{year}{2002}).

\bibitem[{\citenamefont{Wiseman and Vaccaro}(2003)}]{wiseman03}
\bibinfo{author}{\bibfnamefont{H.~M.} \bibnamefont{Wiseman}} \bibnamefont{and}
  \bibinfo{author}{\bibfnamefont{J.~A.} \bibnamefont{Vaccaro}},
  \bibinfo{journal}{Physical review letters} \textbf{\bibinfo{volume}{91}},
  \bibinfo{pages}{097902} (\bibinfo{year}{2003}).

\bibitem[{\citenamefont{Eckert et~al.}(2002)\citenamefont{Eckert, Schliemann,
  Bruss, and Lewenstein}}]{eckert02}
\bibinfo{author}{\bibfnamefont{K.}~\bibnamefont{Eckert}},
  \bibinfo{author}{\bibfnamefont{J.}~\bibnamefont{Schliemann}},
  \bibinfo{author}{\bibfnamefont{D.}~\bibnamefont{Bruss}}, \bibnamefont{and}
  \bibinfo{author}{\bibfnamefont{M.}~\bibnamefont{Lewenstein}},
  \bibinfo{journal}{Annals of Physics} \textbf{\bibinfo{volume}{299}},
  \bibinfo{pages}{88} (\bibinfo{year}{2002}).

\bibitem[{\citenamefont{Amico et~al.}(2008)\citenamefont{Amico, Fazio,
  Osterloh, and Vedral}}]{amico08}
\bibinfo{author}{\bibfnamefont{L.}~\bibnamefont{Amico}},
  \bibinfo{author}{\bibfnamefont{R.}~\bibnamefont{Fazio}},
  \bibinfo{author}{\bibfnamefont{A.}~\bibnamefont{Osterloh}}, \bibnamefont{and}
  \bibinfo{author}{\bibfnamefont{V.}~\bibnamefont{Vedral}},
  \bibinfo{journal}{Reviews of Modern Physics} \textbf{\bibinfo{volume}{80}},
  \bibinfo{pages}{517} (\bibinfo{year}{2008}).

\bibitem[{\citenamefont{Iemini et~al.}(2015{\natexlab{a}})\citenamefont{Iemini,
  Maciel, and Vianna}}]{iemini15}
\bibinfo{author}{\bibfnamefont{F.}~\bibnamefont{Iemini}},
  \bibinfo{author}{\bibfnamefont{T.~O.} \bibnamefont{Maciel}},
  \bibnamefont{and} \bibinfo{author}{\bibfnamefont{R.~O.}
  \bibnamefont{Vianna}}, \bibinfo{journal}{Phys. Rev. B}
  \textbf{\bibinfo{volume}{92}}, \bibinfo{pages}{075423}
  (\bibinfo{year}{2015}{\natexlab{a}}),
  \urlprefix\url{http://link.aps.org/doi/10.1103/PhysRevB.92.075423}.

\bibitem[{\citenamefont{Friis et~al.}(2013)\citenamefont{Friis, Lee, and
  Bruschi}}]{friis13}
\bibinfo{author}{\bibfnamefont{N.}~\bibnamefont{Friis}},
  \bibinfo{author}{\bibfnamefont{A.~R.} \bibnamefont{Lee}}, \bibnamefont{and}
  \bibinfo{author}{\bibfnamefont{D.~E.} \bibnamefont{Bruschi}},
  \bibinfo{journal}{Physical Review A} \textbf{\bibinfo{volume}{87}},
  \bibinfo{pages}{022338} (\bibinfo{year}{2013}).

\bibitem[{\citenamefont{Balachandran et~al.}(2013)\citenamefont{Balachandran,
  Govindarajan, de~Queiroz, and Reyes-Lega}}]{balachandran}
\bibinfo{author}{\bibfnamefont{A.~P.} \bibnamefont{Balachandran}},
  \bibinfo{author}{\bibfnamefont{T.~R.} \bibnamefont{Govindarajan}},
  \bibinfo{author}{\bibfnamefont{A.~R.} \bibnamefont{de~Queiroz}},
  \bibnamefont{and} \bibinfo{author}{\bibfnamefont{A.~F.}
  \bibnamefont{Reyes-Lega}}, \bibinfo{journal}{Phys. Rev. Lett.}
  \textbf{\bibinfo{volume}{110}}, \bibinfo{pages}{080503}
  (\bibinfo{year}{2013}),
  \urlprefix\url{http://link.aps.org/doi/10.1103/PhysRevLett.110.080503}.

\bibitem[{\citenamefont{Iemini et~al.}(2013)\citenamefont{Iemini, Maciel,
  Debarba, and Vianna}}]{iemini13}
\bibinfo{author}{\bibfnamefont{F.}~\bibnamefont{Iemini}},
  \bibinfo{author}{\bibfnamefont{T.~O.} \bibnamefont{Maciel}},
  \bibinfo{author}{\bibfnamefont{T.}~\bibnamefont{Debarba}}, \bibnamefont{and}
  \bibinfo{author}{\bibfnamefont{R.~O.} \bibnamefont{Vianna}},
  \bibinfo{journal}{Quantum Information Processing}
  \textbf{\bibinfo{volume}{12}}, \bibinfo{pages}{733} (\bibinfo{year}{2013}),
  ISSN \bibinfo{issn}{1573-1332},
  \urlprefix\url{http://dx.doi.org/10.1007/s11128-012-0415-6}.

\bibitem[{\citenamefont{Iemini and Vianna}(2013)}]{iemini_ComputableMeasures}
\bibinfo{author}{\bibfnamefont{F.}~\bibnamefont{Iemini}} \bibnamefont{and}
  \bibinfo{author}{\bibfnamefont{R.~O.} \bibnamefont{Vianna}},
  \bibinfo{journal}{Phys. Rev. A} \textbf{\bibinfo{volume}{87}},
  \bibinfo{pages}{022327} (\bibinfo{year}{2013}),
  \urlprefix\url{http://link.aps.org/doi/10.1103/PhysRevA.87.022327}.

\bibitem[{\citenamefont{Henderson and Vedral}(2001)}]{henderson}
\bibinfo{author}{\bibfnamefont{L.}~\bibnamefont{Henderson}} \bibnamefont{and}
  \bibinfo{author}{\bibfnamefont{V.}~\bibnamefont{Vedral}},
  \bibinfo{journal}{Journal of Physics A: Mathematical and General}
  \textbf{\bibinfo{volume}{34}}, \bibinfo{pages}{6899} (\bibinfo{year}{2001}),
  \urlprefix\url{http://stacks.iop.org/0305-4470/34/i=35/a=315}.

\bibitem[{\citenamefont{Ollivier and Zurek}(2001)}]{zurek2001}
\bibinfo{author}{\bibfnamefont{H.}~\bibnamefont{Ollivier}} \bibnamefont{and}
  \bibinfo{author}{\bibfnamefont{W.~H.} \bibnamefont{Zurek}},
  \bibinfo{journal}{Phys. Rev. Lett.} \textbf{\bibinfo{volume}{88}},
  \bibinfo{pages}{017901} (\bibinfo{year}{2001}),
  \urlprefix\url{http://link.aps.org/doi/10.1103/PhysRevLett.88.017901}.

\bibitem[{\citenamefont{Modi et~al.}(2010)\citenamefont{Modi, Paterek, Son,
  Vedral, and Williamson}}]{modiprl2010}
\bibinfo{author}{\bibfnamefont{K.}~\bibnamefont{Modi}},
  \bibinfo{author}{\bibfnamefont{T.}~\bibnamefont{Paterek}},
  \bibinfo{author}{\bibfnamefont{W.}~\bibnamefont{Son}},
  \bibinfo{author}{\bibfnamefont{V.}~\bibnamefont{Vedral}}, \bibnamefont{and}
  \bibinfo{author}{\bibfnamefont{M.}~\bibnamefont{Williamson}},
  \bibinfo{journal}{Phys. Rev. Lett.} \textbf{\bibinfo{volume}{104}},
  \bibinfo{pages}{080501} (\bibinfo{year}{2010}).

\bibitem[{\citenamefont{Debarba et~al.}(2012)\citenamefont{Debarba, Maciel, and
  Vianna}}]{debarbapra}
\bibinfo{author}{\bibfnamefont{T.}~\bibnamefont{Debarba}},
  \bibinfo{author}{\bibfnamefont{T.~O.} \bibnamefont{Maciel}},
  \bibnamefont{and} \bibinfo{author}{\bibfnamefont{R.~O.}
  \bibnamefont{Vianna}}, \bibinfo{journal}{Phys. Rev. A}
  \textbf{\bibinfo{volume}{86}}, \bibinfo{pages}{024302}
  (\bibinfo{year}{2012}).

\bibitem[{\citenamefont{Iemini et~al.}(2014)\citenamefont{Iemini, Debarba, and
  Vianna}}]{iemini14}
\bibinfo{author}{\bibfnamefont{F.}~\bibnamefont{Iemini}},
  \bibinfo{author}{\bibfnamefont{T.}~\bibnamefont{Debarba}}, \bibnamefont{and}
  \bibinfo{author}{\bibfnamefont{R.~O.} \bibnamefont{Vianna}},
  \bibinfo{journal}{Phys. Rev. A} \textbf{\bibinfo{volume}{89}},
  \bibinfo{pages}{032324} (\bibinfo{year}{2014}),
  \urlprefix\url{http://link.aps.org/doi/10.1103/PhysRevA.89.032324}.

\bibitem[{\citenamefont{Piani et~al.}(2011)\citenamefont{Piani, Gharibian,
  Adesso, Calsamiglia, Horodecki, and Winter}}]{pianiprl}
\bibinfo{author}{\bibfnamefont{M.}~\bibnamefont{Piani}},
  \bibinfo{author}{\bibfnamefont{S.}~\bibnamefont{Gharibian}},
  \bibinfo{author}{\bibfnamefont{G.}~\bibnamefont{Adesso}},
  \bibinfo{author}{\bibfnamefont{J.}~\bibnamefont{Calsamiglia}},
  \bibinfo{author}{\bibfnamefont{P.}~\bibnamefont{Horodecki}},
  \bibnamefont{and} \bibinfo{author}{\bibfnamefont{A.}~\bibnamefont{Winter}},
  \bibinfo{journal}{Phys. Rev. Lett.} \textbf{\bibinfo{volume}{106}},
  \bibinfo{pages}{220403} (\bibinfo{year}{2011}),
  \urlprefix\url{http://link.aps.org/doi/10.1103/PhysRevLett.106.220403}.

\bibitem[{\citenamefont{Streltsov et~al.}(2011)\citenamefont{Streltsov,
  Kampermann, and Bru{\ss}}}]{streltsov11}
\bibinfo{author}{\bibfnamefont{A.}~\bibnamefont{Streltsov}},
  \bibinfo{author}{\bibfnamefont{H.}~\bibnamefont{Kampermann}},
  \bibnamefont{and} \bibinfo{author}{\bibfnamefont{D.}~\bibnamefont{Bru{\ss}}},
  \bibinfo{journal}{Physical review letters} \textbf{\bibinfo{volume}{106}},
  \bibinfo{pages}{160401} (\bibinfo{year}{2011}),
  \urlprefix\url{http://prl.aps.org/pdf/PRL/v106/i16/e160401}.

\bibitem[{\citenamefont{Benatti et~al.}(2010)\citenamefont{Benatti, Floreanini,
  and Marzolino}}]{benatti10}
\bibinfo{author}{\bibfnamefont{F.}~\bibnamefont{Benatti}},
  \bibinfo{author}{\bibfnamefont{R.}~\bibnamefont{Floreanini}},
  \bibnamefont{and}
  \bibinfo{author}{\bibfnamefont{U.}~\bibnamefont{Marzolino}},
  \bibinfo{journal}{Annals of Physics} \textbf{\bibinfo{volume}{325}},
  \bibinfo{pages}{924} (\bibinfo{year}{2010}).

\bibitem[{\citenamefont{Ba{\~n}uls et~al.}(2007)\citenamefont{Ba{\~n}uls,
  Cirac, and Wolf}}]{banulus07}
\bibinfo{author}{\bibfnamefont{M.-C.} \bibnamefont{Ba{\~n}uls}},
  \bibinfo{author}{\bibfnamefont{J.~I.} \bibnamefont{Cirac}}, \bibnamefont{and}
  \bibinfo{author}{\bibfnamefont{M.~M.} \bibnamefont{Wolf}},
  \bibinfo{journal}{Physical Review A} \textbf{\bibinfo{volume}{76}},
  \bibinfo{pages}{022311} (\bibinfo{year}{2007}).

\bibitem[{\citenamefont{Brodutch and Modi}(2012)}]{brodutch12}
\bibinfo{author}{\bibfnamefont{A.}~\bibnamefont{Brodutch}} \bibnamefont{and}
  \bibinfo{author}{\bibfnamefont{K.}~\bibnamefont{Modi}},
  \bibinfo{journal}{Quantum Information \& Computation}
  \textbf{\bibinfo{volume}{12}}, \bibinfo{pages}{721} (\bibinfo{year}{2012}).

\bibitem[{\citenamefont{Luo}(2008)}]{fu08}
\bibinfo{author}{\bibfnamefont{S.}~\bibnamefont{Luo}}, \bibinfo{journal}{Phys.
  Rev. A} \textbf{\bibinfo{volume}{77}}, \bibinfo{pages}{022301}
  (\bibinfo{year}{2008}),
  \urlprefix\url{http://link.aps.org/doi/10.1103/PhysRevA.77.022301}.

\bibitem[{\citenamefont{Nakano et~al.}(2012)\citenamefont{Nakano, Piani, and
  Adesso}}]{taka}
\bibinfo{author}{\bibfnamefont{T.}~\bibnamefont{Nakano}},
  \bibinfo{author}{\bibfnamefont{M.}~\bibnamefont{Piani}}, \bibnamefont{and}
  \bibinfo{author}{\bibfnamefont{G.}~\bibnamefont{Adesso}}
  (\bibinfo{year}{2012}), \eprint{1211.4022v1},
  \urlprefix\url{http://arxiv.org/abs/1211.4022v1}.

\bibitem[{\citenamefont{Oppenheim et~al.}(2002)\citenamefont{Oppenheim,
  Horodecki, Horodecki, and Horodecki}}]{opp}
\bibinfo{author}{\bibfnamefont{J.}~\bibnamefont{Oppenheim}},
  \bibinfo{author}{\bibfnamefont{M.}~\bibnamefont{Horodecki}},
  \bibinfo{author}{\bibfnamefont{P.}~\bibnamefont{Horodecki}},
  \bibnamefont{and}
  \bibinfo{author}{\bibfnamefont{R.}~\bibnamefont{Horodecki}},
  \bibinfo{journal}{Phys. Rev. Lett.} \textbf{\bibinfo{volume}{89}},
  \bibinfo{pages}{180402} (\bibinfo{year}{2002}),
  \urlprefix\url{http://link.aps.org/doi/10.1103/PhysRevLett.89.180402}.

\bibitem[{\citenamefont{Horodecki et~al.}(2005)\citenamefont{Horodecki,
  Horodecki, Horodecki, Oppenheim, De, Sen, and
  Synak-Radtke}}]{horodecki2005local}
\bibinfo{author}{\bibfnamefont{M.}~\bibnamefont{Horodecki}},
  \bibinfo{author}{\bibfnamefont{P.}~\bibnamefont{Horodecki}},
  \bibinfo{author}{\bibfnamefont{R.}~\bibnamefont{Horodecki}},
  \bibinfo{author}{\bibfnamefont{J.}~\bibnamefont{Oppenheim}},
  \bibinfo{author}{\bibfnamefont{A.}~\bibnamefont{De}},
  \bibinfo{author}{\bibfnamefont{U.}~\bibnamefont{Sen}}, \bibnamefont{and}
  \bibinfo{author}{\bibfnamefont{B.}~\bibnamefont{Synak-Radtke}},
  \bibinfo{journal}{Physical Review A} \textbf{\bibinfo{volume}{71}},
  \bibinfo{pages}{062307} (\bibinfo{year}{2005}).

\bibitem[{\citenamefont{Ban et~al.}(1997)\citenamefont{Ban, Kurokawa, Momose,
  and Hirota}}]{ban97}
\bibinfo{author}{\bibfnamefont{M.}~\bibnamefont{Ban}},
  \bibinfo{author}{\bibfnamefont{K.}~\bibnamefont{Kurokawa}},
  \bibinfo{author}{\bibfnamefont{R.}~\bibnamefont{Momose}}, \bibnamefont{and}
  \bibinfo{author}{\bibfnamefont{O.}~\bibnamefont{Hirota}},
  \bibinfo{journal}{International Journal of Theoretical Physics}
  \textbf{\bibinfo{volume}{36}}, \bibinfo{pages}{1269} (\bibinfo{year}{1997}).

\bibitem[{\citenamefont{Sasaki et~al.}(1999)\citenamefont{Sasaki, Barnett,
  Jozsa, Osaki, and Hirota}}]{sasaki99}
\bibinfo{author}{\bibfnamefont{M.}~\bibnamefont{Sasaki}},
  \bibinfo{author}{\bibfnamefont{S.~M.} \bibnamefont{Barnett}},
  \bibinfo{author}{\bibfnamefont{R.}~\bibnamefont{Jozsa}},
  \bibinfo{author}{\bibfnamefont{M.}~\bibnamefont{Osaki}}, \bibnamefont{and}
  \bibinfo{author}{\bibfnamefont{O.}~\bibnamefont{Hirota}},
  \bibinfo{journal}{Physical Review A} \textbf{\bibinfo{volume}{59}},
  \bibinfo{pages}{3325} (\bibinfo{year}{1999}),
  \urlprefix\url{http://arxiv.org/pdf/quant-ph/9812062.pdf}.

\bibitem[{\citenamefont{Davies}(1978)}]{davies}
\bibinfo{author}{\bibfnamefont{E.}~\bibnamefont{Davies}},
  \bibinfo{journal}{Information Theory, IEEE Transactions on}
  \textbf{\bibinfo{volume}{24}}, \bibinfo{pages}{596} (\bibinfo{year}{1978}).

\bibitem[{\citenamefont{Chiribella and {Mauro D'Ariano}}(2006)}]{dariano06}
\bibinfo{author}{\bibfnamefont{G.}~\bibnamefont{Chiribella}} \bibnamefont{and}
  \bibinfo{author}{\bibfnamefont{G.}~\bibnamefont{{Mauro D'Ariano}}},
  \bibinfo{journal}{Journal of Mathematical Physics}
  \textbf{\bibinfo{volume}{47}}, \bibinfo{eid}{092107} (\bibinfo{year}{2006}),
  \urlprefix\url{http://scitation.aip.org/content/aip/journal/jmp/47/9/10.1063/1.2349481}.

\bibitem[{\citenamefont{Gigena and Rossignoli}(2015)}]{gigena15}
\bibinfo{author}{\bibfnamefont{N.}~\bibnamefont{Gigena}} \bibnamefont{and}
  \bibinfo{author}{\bibfnamefont{R.}~\bibnamefont{Rossignoli}},
  \bibinfo{journal}{Phys. Rev. A} \textbf{\bibinfo{volume}{92}},
  \bibinfo{pages}{042326} (\bibinfo{year}{2015}),
  \urlprefix\url{http://link.aps.org/doi/10.1103/PhysRevA.92.042326}.

\bibitem[{\citenamefont{Gigena and Rossignoli}(2016)}]{gigena16}
\bibinfo{author}{\bibfnamefont{N.}~\bibnamefont{Gigena}} \bibnamefont{and}
  \bibinfo{author}{\bibfnamefont{R.}~\bibnamefont{Rossignoli}},
  \bibinfo{journal}{Physical Review A} \textbf{\bibinfo{volume}{94}},
  \bibinfo{pages}{042315} (\bibinfo{year}{2016}).

\bibitem[{\citenamefont{Zanardi}(2001)}]{zanardprl}
\bibinfo{author}{\bibfnamefont{P.}~\bibnamefont{Zanardi}},
  \bibinfo{journal}{Phys. Rev. Lett.} \textbf{\bibinfo{volume}{87}},
  \bibinfo{pages}{077901} (\bibinfo{year}{2001}),
  \urlprefix\url{http://link.aps.org/doi/10.1103/PhysRevLett.87.077901}.

\bibitem[{\citenamefont{Iemini et~al.}(2016)\citenamefont{Iemini, Rossini,
  Fazio, Diehl, and Mazza}}]{Iemini_Dissipation_2016}
\bibinfo{author}{\bibfnamefont{F.}~\bibnamefont{Iemini}},
  \bibinfo{author}{\bibfnamefont{D.}~\bibnamefont{Rossini}},
  \bibinfo{author}{\bibfnamefont{R.}~\bibnamefont{Fazio}},
  \bibinfo{author}{\bibfnamefont{S.}~\bibnamefont{Diehl}}, \bibnamefont{and}
  \bibinfo{author}{\bibfnamefont{L.}~\bibnamefont{Mazza}},
  \bibinfo{journal}{Phys. Rev. B} \textbf{\bibinfo{volume}{93}},
  \bibinfo{pages}{115113} (\bibinfo{year}{2016}),
  \urlprefix\url{http://link.aps.org/doi/10.1103/PhysRevB.93.115113}.

\bibitem[{\citenamefont{Diehl et~al.}(2008)\citenamefont{Diehl, Micheli,
  Kantian, Kraus, Buchler, and Zoller}}]{diehl08}
\bibinfo{author}{\bibfnamefont{S.}~\bibnamefont{Diehl}},
  \bibinfo{author}{\bibfnamefont{A.}~\bibnamefont{Micheli}},
  \bibinfo{author}{\bibfnamefont{A.}~\bibnamefont{Kantian}},
  \bibinfo{author}{\bibfnamefont{B.}~\bibnamefont{Kraus}},
  \bibinfo{author}{\bibfnamefont{H.}~\bibnamefont{Buchler}}, \bibnamefont{and}
  \bibinfo{author}{\bibfnamefont{P.}~\bibnamefont{Zoller}},
  \bibinfo{journal}{Nature Phys.} \textbf{\bibinfo{volume}{4}},
  \bibinfo{pages}{878} (\bibinfo{year}{2008}).

\bibitem[{\citenamefont{Verstraete et~al.}(2009)\citenamefont{Verstraete, Wolf,
  and Cirac}}]{verstraete09}
\bibinfo{author}{\bibfnamefont{F.}~\bibnamefont{Verstraete}},
  \bibinfo{author}{\bibfnamefont{M.~M.} \bibnamefont{Wolf}}, \bibnamefont{and}
  \bibinfo{author}{\bibfnamefont{J.~I.} \bibnamefont{Cirac}},
  \bibinfo{journal}{Nature Phys.} \textbf{\bibinfo{volume}{5}},
  \bibinfo{pages}{633} (\bibinfo{year}{2009}).

\bibitem[{\citenamefont{Diehl et~al.}(2011)\citenamefont{Diehl, Rico, Baranov,
  and Zoller}}]{diehl11}
\bibinfo{author}{\bibfnamefont{S.}~\bibnamefont{Diehl}},
  \bibinfo{author}{\bibfnamefont{E.}~\bibnamefont{Rico}},
  \bibinfo{author}{\bibfnamefont{M.}~\bibnamefont{Baranov}}, \bibnamefont{and}
  \bibinfo{author}{\bibfnamefont{P.}~\bibnamefont{Zoller}},
  \bibinfo{journal}{Nature Phys.} \textbf{\bibinfo{volume}{7}},
  \bibinfo{pages}{971} (\bibinfo{year}{2011}).

\bibitem[{\citenamefont{Iemini et~al.}(2015{\natexlab{b}})\citenamefont{Iemini,
  Mazza, Rossini, Fazio, and Diehl}}]{iemini_Maj_H}
\bibinfo{author}{\bibfnamefont{F.}~\bibnamefont{Iemini}},
  \bibinfo{author}{\bibfnamefont{L.}~\bibnamefont{Mazza}},
  \bibinfo{author}{\bibfnamefont{D.}~\bibnamefont{Rossini}},
  \bibinfo{author}{\bibfnamefont{R.}~\bibnamefont{Fazio}}, \bibnamefont{and}
  \bibinfo{author}{\bibfnamefont{S.}~\bibnamefont{Diehl}},
  \bibinfo{journal}{Phys. Rev. Lett.} \textbf{\bibinfo{volume}{115}},
  \bibinfo{pages}{156402} (\bibinfo{year}{2015}{\natexlab{b}}),
  \urlprefix\url{http://link.aps.org/doi/10.1103/PhysRevLett.115.156402}.

\end{thebibliography}
\end{document}